\documentclass[reprint,onecolumn,aps,amsmath,amssymb,pra,superscriptaddress,floatfix,tightenlines,11pt]{revtex4-2}

\usepackage{amsmath, amssymb, amsthm}
\usepackage{mathrsfs}
\usepackage{mathtools}
\usepackage{braket}
\usepackage{bbm}
\usepackage{bm}
\usepackage{color}
\usepackage[dvipsnames]{xcolor}
\usepackage[most]{tcolorbox}
\usepackage{comment}
\usepackage{hyperref}
\usepackage{url}
\usepackage{cleveref}
\usepackage[T1]{fontenc}
\usepackage{times,txfonts}

\usepackage{tikz}
\usetikzlibrary{quantikz2}

\hypersetup{
    bookmarksnumbered=true,
    unicode=false,
    pdfstartview={FitH},
    pdftitle={},
    pdfauthor={},
    pdfsubject={},
    pdfcreator={},
    pdfproducer={},
    pdfkeywords={},
    pdfnewwindow=true,
    colorlinks=true,
    linkcolor=NavyBlue,
    citecolor=NavyBlue,
    filecolor=NavyBlue,
    urlcolor=NavyBlue
}

\DeclareMathOperator{\tr}{Tr}
\DeclareMathOperator{\Tr}{Tr}

\DeclareMathOperator{\ran}{range}
\DeclareMathOperator{\Span}{span}
\DeclareMathOperator{\supp}{supp}

\DeclareMathOperator{\id}{\mathrm{id}}

\DeclareMathOperator{\centropy}{H}
\DeclareMathOperator{\qentropy}{S}
\DeclareMathOperator{\fidelity}{F}
\DeclareMathOperator{\Lin}{L}

\DeclareMathOperator{\Channel}{C}
\DeclareMathOperator{\Unitary}{U}
\DeclareMathOperator{\Density}{D}

\newcommand{\opnorm}[1]{\left\|#1\right\|_\infty}

\newcommand{\vk}[1]{\lvert #1\rangle\!\rangle}
\newcommand{\vb}[1]{\langle\!\langle #1\rvert}
\newcommand{\vkb}[2]{\lvert #1 \rangle\!\rangle \langle \!\langle #2 \rvert}
\newcommand{\vbk}[1]{\langle\!\langle #1 \rangle\!\rangle}

\newcommand{\cE}{\mathcal E}
\newcommand{\cD}{\mathcal D}
\newcommand{\cH}{\mathcal H}
\newcommand{\cT}{\mathcal T}

\newcommand{\Trans}{\mathsf{T}}
\newcommand{\Prob}{\mathcal{P}}

\newcommand{\ketbra}[2]{|#1\rangle\langle#2|}

\allowdisplaybreaks[4]

\theoremstyle{plain}
\newtheorem{theorem}{Theorem}
\newtheorem{lemma}[theorem]{Lemma}
\newtheorem{corollary}[theorem]{Corollary}
\newtheorem{proposition}[theorem]{Proposition}
\theoremstyle{definition}
\newtheorem{definition}[theorem]{Definition}

\theoremstyle{remark}
\newtheorem{conjecture}{Conjecture}
\newtheorem{remark}[conjecture]{Remark}

\begin{document}
\author{Kohdai Kuroiwa}
\affiliation{Institute for Quantum Computing and Department of Combinatorics and Optimization, University of Waterloo, Ontario, Canada, N2L 3G1}
\affiliation{Perimeter Institute for Theoretical Physics, Ontario, Canada, N2L 2Y5}

\title{Exponential strong converse for blind quantum data compression}

\begin{abstract}
We establish exponential strong converse theorems for blind quantum data compression of finite-dimensional mixed-state sources, with and without entanglement assistance.
Blind compression is a fundamental quantum information processing task in which the sender aims to transmit quantum states without direct access to the classical label.
Strong converse is a notable aspect of data compression: the accuracy of the compression abruptly drops down to zero once the compression rate goes below a threshold. 
Strong converses for data compression have been established for pure-state sources; however, for mixed-state sources, 
the optimal compression rate is determined by a structural decomposition of quantum states and can change substantially under small perturbations.
This sensitivity makes it difficult to quantitatively analyze
the trade-off between compression rate and error.
Despite recent progress, whether a strong converse holds at the optimal rate has remained unresolved in general.
In this work, we resolve this question affirmatively.
For unassisted compression, we show that at every rate below the optimum, the accuracy decays exponentially with the block length for general mixed-state sources without restrictions on protocols.
We also prove exponential strong converse for blind compression with unlimited entanglement assistance at its optimal rate. 
Combining the two strong converse bounds yields exponential decay at every compression-entanglement rate pair outside the complete achievable rate region. 
Our proof introduces a new overlap quantity based on the structure of quantum states; this quantity connects preservation of the source information to the dimension of the transmitted system.
We expect that this approach provides useful tools for converse analyses of other quantum information-processing tasks involving mixed states.
\end{abstract}

\maketitle

\section{Introduction}~\label{sec:introduction}
\emph{Data compression} is a fundamental task in information theory \cite{Shannon1948}.
In this task, a \emph{sender} (Alice) aims to send data to a \emph{receiver} (Bob) accurately and efficiently by appropriately encoding and decoding the data.
Data is drawn from a known identical probability distribution independently an asymptotically large number of times.
The goal is to transmit the data as efficiently as possible with asymptotically vanishing error.
The minimum achievable number of bits per draw is called the optimal compression rate.

\emph{Quantum data compression} generalizes classical data compression to the quantum setting.
In quantum data compression, a \emph{referee} draws a label from a known probability distribution, records the label on a classical system, and repeats this procedure independently with the same distribution an asymptotically many times.
The referee then prepares quantum states according to the labels.
The resulting joint classical-quantum state is called a \emph{quantum source}.
The referee gives the quantum source systems to Alice.
Alice then encodes the quantum states into a number of qubits and transmits them to Bob, who decodes the received qubits to reconstruct the source states.

In this work, we study \emph{blind compression}, in which Alice receives the quantum source system but has no direct access to the classical system containing its preparation label.
We consider both \emph{unassisted compression}, in which Alice and Bob share no additional resources, and \emph{entanglement-assisted compression}, 
in which they may use preshared entanglement.
In \emph{visible compression}~\cite{Horodecki2000,Hayashi2006,Bennett2005,Bennett2014a},
which is not the focus of this work, the referee also provides Alice with the classical preparation label.

Early studies of blind quantum data compression focused on
the unassisted compression of \emph{pure-state ensemble sources},
in which every state in the source ensemble is
pure~\cite{Schumacher1995,Richard1994,Barnum1996}. 
For general mixed-state ensemble sources, Refs.~\cite{Horodecki1998,Barnum2001} established a lower bound on the unassisted compression rate, which coincides with the optimal rate for pure-state sources.
Studies of commuting mixed states also highlighted differences between visible and blind compression~\cite{Dur2001,Kramer2001}.
Koashi and Imoto~\cite{Koashi2001a} subsequently determined the optimal rate of unassisted blind compression using a structural decomposition of quantum states known as the \emph{Koashi--Imoto (KI) decomposition}, or \emph{KI structure}~\cite{Koashi2002}.
Entanglement-assisted compression was later studied by
Baghali Khanian and Winter for pure-state
ensembles~\cite{Khanian2019}.
They subsequently characterized the optimal communication-entanglement rate region for general mixed-state sources with a quantum reference, including blind compression of mixed-state ensembles as a special case~\cite{Khanian2020b}.

The trade-off between compression rate and error has attracted considerable interest. 
Its behavior depends on the error criterion used.
Two error criteria have received particular attention: \emph{global error} (\emph{block error}) and \emph{local error} (\emph{per-copy error}).
Global error quantifies the discrepancy over the entire block, whereas local error quantifies discrepancies in individual copies.
The global-error criterion imposes a stricter accuracy requirement than the local-error criterion and is often used to establish optimal compression rates.
On the other hand, allowing even a small local error can substantially reduce the required compression rate.
For example, it was pointed out that the rate might be sensitive even to a small error in some cases~\cite{Anshu2022}, and a substantial rate reduction is possible by introducing a small approximation~\cite{Kuroiwa2022}.
More generally, the trade-off between compression rate and local error is studied in \emph{quantum rate-distortion theory}~\cite{Barnum2000, Devetak2002, Datta2013b, Datta2013c, Khanian_Kuroiwa2022}.

\emph{Strong converse} gives a sharp asymptotic limitation on compression under global error: at every rate below the threshold, often referred to as the \emph{strong converse rate}, the accuracy of compression approaches zero as the block length increases.
The early strong converse proofs for pure-state sources assumed unitary decoding~\cite{Schumacher1995,Richard1994}.
Ref.~\cite{Barnum1996} removed this restriction, and Winter~\cite{Winter1999} subsequently gave another proof.
For mixed-state sources, the sensitivity of the KI structure to approximation makes the problem much more delicate.
Recently, Baghali Khanian~\cite{Khanian2022} established a pretty strong converse for a class of sources called partially exchangeable states and presented a strong-converse bound for unassisted blind compression under additional decoder restrictions.
Despite this recent progress, these results leave open whether the optimal rates obtained by Koashi and Imoto~\cite{Koashi2001a} and by Baghali Khanian and Winter~\cite{Khanian2020b} serve as the strong converse rates
for general mixed-state sources with unrestricted encoding and decoding.

In this paper, we answer this question by establishing an exponential strong converse for blind compression.
First, for the unassisted setting, we prove a strong converse at the optimal Koashi--Imoto rate.
At every rate strictly below this threshold, the compression accuracy decays exponentially with the block length. 
Our proof introduces a new measure, the \emph{KI overlap}, which quantifies preservation of the essential information identified by the KI decomposition.
We relate the KI overlap to the compression accuracy and bound it in terms of the dimension of the compressed system.
The argument has two parts: below the optimal rate, the overlap associated with preserving this information on nearly the entire block is exponentially small; the remaining contributions to the compression accuracy, which correspond to information loss on many copies, are also exponentially suppressed.
Together, these estimates establish the unassisted exponential strong converse without restrictions on the encoder or decoder.
Thus, we connect the KI structure to quantitative bounds on approximate source preservation.
In addition, we establish an exponential strong converse for unlimited entanglement assistance at the optimal quantum rate by deriving an alternative upper bound on the KI overlap. 
Furthermore, combining this entanglement-assisted result with the unassisted strong converse gives exponential decay outside the complete achievable qubit-ebit region identified in Ref.~\cite{Khanian2020b}.
Note that these strong converse results also hold for the more general mixed-state sources introduced in Ref.~\cite{Khanian2020b}.

These results sharpen the fundamental limits of blind quantum data compression under the global-error criterion: allowing any fixed positive global error tolerance below one does not enlarge the achievable rate region.
Our results also clarify the distinction between global and local errors, since rate reductions permitted by local approximation do not imply corresponding reductions under global error.

The rest of this paper is organized as follows.
First, in Sec.~\ref{sec:preliminaries}, we introduce our notation, review the KI decomposition, and formulate the compression model and review its rate region. 
Next, in Sec.~\ref{sec:strong_converse_tools}, we establish the source reductions and introduce the perturbed source and KI overlap used in strong converse proofs.
Then, in Sec.~\ref{sec:main_results}, we prove the unassisted exponential strong converse. 
In Sec.~\ref{sec:assisted}, we establish the strong converse with unlimited entanglement assistance and combine the two results to obtain the complete strong-converse rate region.
Finally, in Sec.~\ref{sec:conclusion}, we summarize the results and discuss open questions.

\section{Preliminaries}
\label{sec:preliminaries}
In this section, we review background material for this paper. 
In Sec.~\ref{subsec:notation}, we introduce our notation. 
In Sec.~\ref{subsec:structure_quantum_ensemble}, we define quantum ensembles and review the Koashi--Imoto (KI) decomposition.
In Sec.~\ref{subsec:blind_compression}, we give an overview of the formal setting of blind quantum data compression and discuss its rate region.

\subsection{Notation}
\label{subsec:notation}
All alphabets in this paper are finite, and all Hilbert spaces are finite-dimensional.
We let $\Prob(\Sigma)$ denote the set of probability distributions on an alphabet $\Sigma$.
We use capital letters $A,B,C,\ldots$ for quantum systems, and the complex Hilbert space of a system $A$ is denoted by $\cH_A$ with dimension $\dim A$.
A vector $\ket{\psi} \in \cH_A$ is called a \textit{pure state} or simply a \textit{state} if it has unit norm: $\|\ket{\psi}\| = 1$.

We let $\Lin(\cH_A)$ denote the set of linear operators on $A$.
For example, $I_{A} \in \Lin(\cH_A)$ represents the identity operator on a system $A$.
We denote the trace of an operator $X$ by $\Tr X$.
The partial trace over a subsystem $A$ is denoted by $\Tr_A$.

A linear operator $H$ on a system $A$ is said to be \textit{Hermitian} if $H^\dagger = H$, where $\dagger$ represents the adjoint. 
A Hermitian operator $P$ on a system $A$ is \textit{positive semidefinite} if $\braket{\psi|P|\psi} \geq 0$ for all vectors $\ket{\psi} \in \cH_A$, and we write $P\geq 0$; it is positive definite if $\braket{\psi|P|\psi} > 0$ for all nonzero vectors $\ket{\psi} \in \cH_A$, and we write $P>0$.
A positive semidefinite operator is called a \textit{density operator} if it has unit trace, and
$\Density(\cH_A)$ denotes the set of density operators on $A$.

When a density operator $\rho \in \Density(\cH_A)$ has rank $1$, it can be written as $\rho = \ketbra{\psi}{\psi}$ using some state $\ket{\psi} \in \cH_A$.
Such a density operator will also be called a pure state.
The spectral decomposition expresses any density operator as a convex combination of pure states, so it can be interpreted as a probabilistic mixture of pure states.
Thus, a density operator also represents a quantum state.
If a density operator $\rho$ is not pure, it is called a \textit{mixed state}.

Transposes and vectorization are taken in a fixed orthonormal basis:
\begin{equation}
X=\sum_{a,b}X_{ab}|a\rangle\langle b|,
\qquad
X^{\mathsf T}=\sum_{a,b}X_{ab}|b\rangle\langle a|,
\qquad
|X\rangle\!\rangle=\sum_{a,b}X_{ab}|a\rangle\otimes|b\rangle.
\end{equation}
By definition, for vectors $\ket{u}$ and $\ket{v}$,
\begin{equation}
    \vk{\ketbra{u}v} = \ket{u} \otimes \overline{\ket{v}},
\end{equation}
where $\overline{\ket{\cdot}}$ denotes complex conjugation with respect to the fixed basis.
For linear operators $A$, $B$, and $V$ of compatible dimensions, vectorization satisfies
\begin{equation}
\label{eq:vec_property}
(A\otimes B^{\mathsf T})|V\rangle\!\rangle
=|AVB\rangle\!\rangle
\end{equation}
and for linear operators $V$ and $W$,
\begin{equation}
\label{eq:vec_trace}
\langle\!\langle V|W\rangle\!\rangle
=\Tr(V^\dagger W).
\end{equation}
Consequently,
\begin{equation}
    \label{eq:vec_braket}
    \langle\!\langle V|A\otimes B^{\mathsf T}|V\rangle\!\rangle=\Tr(V^\dagger AVB).
\end{equation}
See also Sec.~1.1.2 of Ref.~\cite{Watrous2018}.

An isometry $U$ from $\cH_A$ to $\cH_B$ satisfies
\begin{equation}
    U^\dagger U=I_A.
\end{equation}
We denote the set of such isometries by $\Unitary(\cH_A,\cH_B)$.
When $\cH_A \simeq \cH_B$, $U$ is said to be unitary.

Quantum channels are completely positive trace-preserving linear maps.
We denote the set of quantum channels that map $\Lin(\cH_A)$ to $\Lin(\cH_B)$ by
$\Channel(\cH_A,\cH_B)$,
and write $\id_A$ for the identity channel on $A$.
In this work, we write $\mathcal{N}:A \to B$ for $\mathcal{N} \in \Channel(\cH_A,\cH_B)$.
The adjoint $\mathcal N^\dagger$ of a linear map $\mathcal N$ is defined by
\begin{equation}
    \Tr[Y^\dagger\mathcal N(X)] =\Tr[(\mathcal N^\dagger(Y))^\dagger X].
\end{equation}
In particular, the adjoint of a trace-preserving map $\mathcal{N} \in \Channel(\cH_A, \cH_B)$ is unital: $\mathcal N^\dagger (I_B) = I_A$.

For a linear map $\mathcal N$ from $\Lin(\cH_A)$ to $\Lin(\cH_B)$, 
\begin{equation}
 J_{\mathcal N}\coloneqq (\mathcal N\otimes\id)(\vk{I}\vb{I}) \in \Lin(\cH_B\otimes \cH_A)
\end{equation}
denotes its unnormalized Choi operator. 
For $X \in \Lin(\cH_A)$ and $Y \in \Lin(\cH_B)$, it holds that 
\begin{equation}
    \label{eq:choi_relation}
    \Tr[J_{\mathcal N}(Y\otimes X^{\Trans})]
    =\Tr[Y\mathcal N(X)].
\end{equation}
When $\mathcal N$ is a channel, $J_{\mathcal N} \geq 0$ and $\tr_B(J_{\mathcal N}) = I_A$. 

The trace, Hilbert--Schmidt, and operator norms are, respectively,
\begin{equation}
    \|X\|_1=\Tr\sqrt{X^\dagger X},
    \qquad
    \|X\|_2=\sqrt{\Tr(X^\dagger X)},
    \qquad
    \|X\|_\infty=s_{\max}(X),
\end{equation}
where $s_{\max}$ denotes the largest singular value.
For every quantum channel $\mathcal N \in \Channel(\cH_A,\cH_B)$
and every linear operator $X \in \Lin(\cH_A)$,
\begin{equation}
    \|\mathcal N(X)\|_1\leq\|X\|_1.
\end{equation}
For positive semidefinite operators $P$ and $Q$,
the (squared) fidelity is
\begin{equation}
    F(P,Q)=\|\sqrt P\sqrt Q\|_1^2.
\end{equation}
For $P,Q \geq 0$ and $Z > 0$,
H\"older's inequality~\cite[Eq.~(1.174)]{Watrous2018} gives
\begin{equation}
    \label{eq:fidelity_trace}
    \fidelity(P,Q)\leq
\|\sqrt P Z^{-1/2}\|_2^2\|Z^{1/2}\sqrt Q\|_2^2
=\Tr(PZ^{-1})\Tr(QZ).
\end{equation}

For a probability vector $p$ on $\Sigma$, the Shannon entropy $\centropy(p)$ is defined by
\begin{equation}
    \centropy(p) \coloneqq \sum_{x \in \Sigma} -p_x \log_2p_x.
\end{equation}
The binary entropy function $h_2:[0,1] \to [0,1]$ is defined by
\begin{equation}
    \label{eq:binary_entropy}
    h_2(\theta) \coloneqq \centropy((\theta,1-\theta)) = -\theta\log_2\theta-(1-\theta)\log_2(1-\theta).
\end{equation}
For a quantum state $\rho$, the quantum entropy (von Neumann entropy) is defined by
\begin{equation}
    \qentropy(\rho) \coloneqq -\tr(\rho\log_2\rho).
\end{equation}

\subsection{Structure of quantum ensembles}\label{subsec:structure_quantum_ensemble}
We first review the definition of quantum ensembles.
\begin{definition}[Quantum ensembles]
Let $A$ be a quantum system with Hilbert space $\cH_A$, and let $\Sigma$ be an alphabet.
A quantum ensemble $\Phi$ on $A$ is a family of pairs, each consisting of a positive probability and a quantum state,
\begin{equation}
    \Phi \coloneqq \{(p_x,\rho_x) \in (0,1]\times\Density(\cH_A): x\in\Sigma\},
\end{equation}
where $p \in \Prob(\Sigma)$.
We often write
\begin{equation}
    \Phi = \{p_x,\rho_x\}_{x\in\Sigma}.
\end{equation}
The average state $\rho_{\Phi}$ of a quantum ensemble $\Phi = \{p_x,\rho_x\}_{x\in\Sigma}$ is defined as
\begin{equation}
    \rho_{\Phi} \coloneqq \sum_{x\in\Sigma} p_x\rho_x.
\end{equation}
\end{definition}

Throughout this paper, we discard labels of zero probability, in accordance with the definition above, and hence assume $p_x>0$ for every $x\in\Sigma$.

Koashi and Imoto identified a structural decomposition of quantum ensembles, known as the \emph{Koashi--Imoto (KI) decomposition} or the \emph{KI structure}~\cite{Koashi2002}.
Intuitively, when we have a quantum ensemble, we can decompose each state in the ensemble into the following three parts: a \emph{classical part}, a \emph{non-redundant quantum part}, and a \emph{redundant part}.
The following theorem provides a formal statement.

\begin{theorem}[KI decomposition~\cite{Koashi2002}]~\label{thm:KIdecomp}
Let $\Phi=\{p_x,\rho_x\}_{x\in\Sigma}$ be an ensemble on a system $A$.
Here, without loss of generality, we may assume $\cH_A=\supp\rho_\Phi$ as shown in Sec.~\ref{subsec:reduction}.
Then, there exists a decomposition of its Hilbert space
\begin{equation}
    \cH_A \simeq \bigoplus_{c \in \Xi} \cH_{Q_c} \otimes \cH_{R_c}
\end{equation}
and a corresponding isometry
\begin{equation}
\Gamma_{\Phi} \in \Unitary\left(\cH_A,\bigoplus_{c \in \Xi} \cH_{Q_c} \otimes \cH_{R_c}\right)
\end{equation}
satisfying the following conditions.
\begin{enumerate}
    \item
    For all $x\in\Sigma$,
    \begin{equation}\label{eq:KIdecomp}
        \Gamma_{\Phi} \rho_x \Gamma_{\Phi}^\dagger = \bigoplus_{c \in \Xi} q_{c|x}\rho^{(x,c)}_{Q_c}\otimes \rho^{(c)}_{R_c}.
    \end{equation}
    Here, for each $x \in \Sigma$,
    $\{q_{c|x}:c \in \Xi\}$ forms a probability distribution over labels $c\in\Xi$; 
    $\rho^{(x,c)}_{Q_c} \in \Density(\cH_{Q_c})$ is a density operator on $Q_c$ that may depend on both $x\in\Sigma$ and $c\in\Xi$, whereas $\rho^{(c)}_{R_c} \in \Density(\cH_{R_c})$ is a density operator on $R_c$ that is independent of $x\in\Sigma$.

    \item
    For each $c\in\Xi$,
    if a projection operator $P: \cH_{Q_c} \to \cH_{Q_c}$ satisfies
    \begin{equation}
        P q_{c|x}\rho^{(x,c)}_{Q_c} = q_{c|x}\rho^{(x,c)}_{Q_c} P
    \end{equation}
    for all $x \in \Sigma$, then $P = I_{\cH_{Q_c}}$ or $P = 0$.

    \item
    For all $c,c^\prime \in \Xi$ such that $c\neq c^\prime$,
    there exists no isometry $V \in \Unitary(\cH_{Q_c},\cH_{Q_{c^\prime}})$ such that
    \begin{equation}
        V q_{c|x}\rho^{(x,c)}_{Q_c} = \alpha q_{c^\prime|x}\rho^{(x,c^\prime)}_{Q_{c^\prime}} V
    \end{equation}
    with some positive real number $\alpha$ for all $x\in\Sigma$.
\end{enumerate}
\end{theorem}

The first statement~\eqref{eq:KIdecomp} expresses every state in $\Phi$ in a common block-diagonal form. 
Each block is called a \textit{KI block}. 
The second and third statements ensure that the decomposition~\eqref{eq:KIdecomp} is maximal; that is, we cannot further refine the structure. Indeed, the second one states that we cannot further decompose an individual block of the KI decomposition; the third rules out equivalence between distinct blocks, up to an isometry and a positive rescaling of their weights.
Note that $\rho^{(c)}_{R_c} \in \Density(\cH_{R_c})$ does not depend on $x\in\Sigma$; thus it is called redundant because it does not contain information about the label $x \in \Sigma$.
Hereafter, when we specify the KI decomposition of a given ensemble $\Phi = \{p_x,\rho_x\}_{x\in\Sigma}$, we may omit $\Gamma_{\Phi}$ and write
\begin{equation}
    \rho_x = \bigoplus_{c \in \Xi} q_{c|x}\rho^{(x,c)}_{Q_c}\otimes \rho^{(c)}_{R_c}
\end{equation}
for brevity.

The KI decomposition defines two quantum channels that remove and reattach the ensemble's redundant parts, respectively.

\begin{proposition}[KI operations~\cite{Koashi2002,Kuroiwa2022}]~\label{prop:KIoperations}
Let $\Phi = \{p_x,\rho_x\}_{x\in\Sigma}$ be an ensemble on a Hilbert space $\cH$.
Consider the KI decomposition of $\Phi$:
\begin{equation}
    \cH \coloneqq \bigoplus_{c \in \Xi} \cH_{Q_c} \otimes \cH_{R_c}
\end{equation}
such that
\begin{equation}
    \rho_x = \bigoplus_{c \in \Xi} q_{c|x}\rho^{(x,c)}_{Q_c}\otimes \rho^{(c)}_{R_c} \quad \forall x \in \Sigma.
\end{equation}
Then, there exist quantum channels $\mathcal{K}_{\mathrm{off}, \Phi}$ and $\mathcal{K}_{\mathrm{on}, \Phi}$ such that
\begin{align}
    &\mathcal{K}_{\mathrm{off}, \Phi}(\rho_x) = \bigoplus_{c \in \Xi} q_{c|x}\rho^{(x,c)}_{Q_c}, \\
    &\mathcal{K}_{\mathrm{on}, \Phi}\left(\bigoplus_{c\in \Xi} q_{c|x}\rho^{(x,c)}_{Q_c}\right) = \rho_x
\end{align}
for all $x\in\Sigma$.
\end{proposition}
The channels $\mathcal{K}_{\mathrm{off},\Phi}$ and $\mathcal{K}_{\mathrm{on},\Phi}$, which discard and reattach the redundant parts, respectively, are called \textit{KI} operations.
Both KI operations first remove coherence between distinct blocks;
$\mathcal{K}_{\mathrm{off},\Phi}$ then takes the partial trace over the redundant subsystem, while $\mathcal{K}_{\mathrm{on},\Phi}$ prepares the corresponding redundant state within each block.
See Ref.~\cite{Kuroiwa2022} for an explicit construction of the KI operations.
In particular, we have
\begin{equation}
    \mathcal{K}_{\mathrm{on},\Phi}\circ \mathcal{K}_{\mathrm{off},\Phi}(\rho_x)=\rho_x
    \quad\text{for every }x.
\end{equation}
Note that these channels are reversible on the source family, but reversibility on all operators is not asserted.

\subsection{Blind quantum data compression}
\label{subsec:blind_compression}
Blind quantum data compression is a communication task between a sender, Alice, and a receiver, Bob.
Alice receives quantum systems independently drawn from a known ensemble by a referee and transmits a compressed quantum system to Bob, who reconstructs the source.
The defining restriction is that Alice has no direct access to the classical system containing the label.

Let $\Phi=\{p_x,\rho_x\}_{x\in\Sigma}$ be an ensemble on $A$.
The referee records the preparation label in a classical system $X$ and prepares the source state
\begin{equation}
 \rho_{XA}\coloneqq
 \sum_{x\in\Sigma}p_x\ketbra{x}{x}_X\otimes\rho_x,
 \qquad
 \rho_A=\rho_\Phi=\sum_{x\in\Sigma}p_x\rho_x.
\end{equation}
For a block of $n$ independent copies, write
\begin{equation}
 \rho_{X^nA^n}\coloneqq\rho_{XA}^{\otimes n}
 =\sum_{x^n\in\Sigma^n}p_{x^n}
 \ketbra{x^n}{x^n}_{X^n}\otimes\rho_{x^n},
\end{equation}
where tensor factors are reordered as necessary, and
\begin{equation}
 p_{x^n}\coloneqq\prod_{i=1}^np_{x_i},
 \qquad
 \ket{x^n}\coloneqq\bigotimes_{i=1}^n\ket{x_i},
 \qquad
 \rho_{x^n}\coloneqq\bigotimes_{i=1}^n\rho_{x_i}.
\end{equation}

Alice and Bob are allowed to share an arbitrary finite-dimensional state $\zeta_n\in\Density(\cH_{A_0}\otimes\cH_{B_0})$, independent of the source.
A compression protocol consists of this state and encoding and decoding channels
\begin{equation}
 \cE_n:A^nA_0\longrightarrow M_n,
 \qquad
 \cD_n:M_nB_0\longrightarrow A^n.
\end{equation}
Alice sends the message system $M_n$ to Bob through a noiseless quantum channel.
For a fixed shared state $\zeta_n$, we also refer to the pair $(\cE_n,\cD_n)$ as the \textit{compression protocol}.
The effective channel acting on the source is
\begin{equation}
 \label{eq:effective-source-channel}
 \cT_n(O)\coloneqq
 \cD_n\!\left[
   (\cE_n\otimes\id_{B_0})(O\otimes\zeta_n)
 \right],
 \qquad O\in\Lin(\cH_A^{\otimes n}),
\end{equation}
and we also refer to $\cT_n$ as the effective protocol channel.
The quantum communication (qubit) rate of the protocol is
\begin{equation}
 r_n\coloneqq\frac{1}{n}\log_2\dim M_n,
\end{equation}
which we call the \textit{compression rate}.

We evaluate accuracy using the global squared fidelity
\begin{equation}
 \label{eq:error_criterion}
 F_n(\Phi,\cT_n)\coloneqq
 \fidelity\!\left(
   (\id_{X^n}\otimes\cT_n)(\rho_{X^nA^n}),
   \rho_{X^nA^n}
 \right)=1-\epsilon_n.
\end{equation}
We call $F_n(\Phi,\cT_n)$ the compression fidelity and write $F_n$ when the source and protocol are clear.
The quantity $\epsilon_n$ is the global error; only this error criterion is used below.

There are three compression scenarios relevant to this paper.
First, in \emph{unassisted compression}, $A_0$ and $B_0$ are one-dimensional; \textit{i.e.,} Alice and Bob do not share anything beforehand, and we have $\cT_n=\cD_n\circ\cE_n$.
Next, in the model with \emph{unlimited entanglement assistance}, the shared state $\zeta_n$ is unrestricted as long as it is finite-dimensional and independent of the source.
Finally, for the \emph{qubit-ebit rate region}, we consider a maximally entangled state as the shared state:
\begin{equation}
 \ket{\Upsilon_{d_n}}\coloneqq
 \frac{1}{\sqrt{d_n}}\sum_{j=1}^{d_n}
 \ket{j}_{A_0}\otimes\ket{j}_{B_0},
 \qquad
 \Upsilon_{d_n}\coloneqq
 \ketbra{\Upsilon_{d_n}}{\Upsilon_{d_n}},
\end{equation}
with
$\cH_{A_0}\simeq\cH_{B_0}\simeq\mathbb C^{d_n}$.
The entanglement (ebit) rate is
\begin{equation}
 e_n\coloneqq\frac{1}{n}\log_2d_n.
\end{equation}

For $R,E\geq0$ and $\epsilon\in[0,1]$, an $(n,R,E,\epsilon)$ code in the qubit-ebit model consists of $\Upsilon_{d_n}$ and encoding and decoding channels satisfying
\begin{equation}
 \log_2\dim M_n\leq nR,
 \qquad
 \log_2d_n\leq nE,
 \qquad
 F_n(\Phi,\cT_n)\geq1-\epsilon.
\end{equation}
A pair $(R,E)$ is \emph{achievable} if, for every
$\epsilon>0$ and $\delta>0$, there exists $n_0$ such that an
$(n,R+\delta,E+\delta,\epsilon)$ code exists for every $n\geq n_0$.
In the unassisted and unlimited-assistance models, achievability is defined analogously using only the compression rate and the corresponding restriction on the shared resource.
The infimum of achievable communication rates is called the \emph{optimal rate}.

We now express the optimal rates in terms of the KI decomposition of $\Phi$.
Define the average probability of block $c \in \Xi$ and its
normalized, redundancy-free average state by
\begin{equation}
 s_{\Phi,c}\coloneqq\sum_{x\in\Sigma}p_xq_{c|x},
 \qquad
 \overline\rho_{\Phi,c}\coloneqq
 \frac{1}{s_{\Phi,c}}
 \sum_{x\in\Sigma}p_xq_{c|x}\rho_{Q_c}^{(x,c)}.
\end{equation}
Write $s_\Phi=(s_{\Phi,c})_{c\in\Xi}$.
After removing the redundant part, we have
\begin{equation}
 \mathcal{K}_{\mathrm{off},\Phi}(\rho_\Phi)
 =\bigoplus_{c\in\Xi}s_{\Phi,c}\overline\rho_{\Phi,c}.
\end{equation}
We define the rate thresholds by
\begin{align}
 R_{\mathrm{blind},\Phi}
 &\coloneqq\qentropy(\mathcal{K}_{\mathrm{off},\Phi}(\rho_\Phi))
 =\centropy(s_\Phi)
  +\sum_{c\in\Xi}s_{\Phi,c}\qentropy(\overline\rho_{\Phi,c}),
 \\
 R_{\mathrm{blind},\Phi}^{\mathrm{EA}}
 &\coloneqq\frac12\centropy(s_\Phi)
  +\sum_{c\in\Xi}s_{\Phi,c}\qentropy(\overline\rho_{\Phi,c})
 =R_{\mathrm{blind},\Phi}-\frac12\centropy(s_\Phi).
\end{align}

\begin{theorem}[Blind compression rate region~\cite{Koashi2001a,Khanian2020b}]
 \label{thm:optimal_blind_compression_rate}
 Let $\Phi$ be a quantum ensemble.
 In the qubit-ebit model, a pair $(R,E)\in\mathbb R_{\geq0}^2$ is achievable if and only if
 \begin{equation}
  \label{eq:rate_region}
  \begin{aligned}
   R+E&\geq R_{\mathrm{blind},\Phi},\\
   R&\geq R_{\mathrm{blind},\Phi}^{\mathrm{EA}}.
  \end{aligned}
 \end{equation}
 In particular, the optimal rate of unassisted compression is $R_{\mathrm{blind},\Phi}$.
 With unlimited maximally entangled assistance, the optimal rate is $R_{\mathrm{blind},\Phi}^{\mathrm{EA}}$.
\end{theorem}

\section{Source reductions and KI overlap}
\label{sec:strong_converse_tools}
In this section, we introduce mathematical tools used in our unassisted and entanglement-assisted strong converse proofs.
In Sec.~\ref{subsec:reduction}, we discuss source reductions that allow the strong converses to be proved for redundancy-free ensembles with positive-definite source states.
In Sec.~\ref{subsec:perturbed-source}, we introduce the perturbed source and the KI overlap used in our proofs, and state a lemma to turn overlap bounds into bounds on compression fidelity.

\subsection{Reduction to redundancy-free positive-definite sources}
\label{subsec:reduction}
In our strong-converse proofs, it suffices to consider ensembles with no redundant subsystem and with positive-definite source states.
Here, we give the reductions explicitly, including their action on protocols.

\paragraph{Restriction to source support.}
Let $A$ be a system whose Hilbert space contains $\supp\rho_\Phi$.
Since $p_x>0$ for every retained label,
\begin{equation}
 \supp\rho_x\subseteq\supp\rho_\Phi
 \qquad\text{for every }x\in\Sigma.
\end{equation}
Let $A_s$ denote the system with $\cH_{A_s}=\supp\rho_\Phi$, let $P$ be the projector onto this support, and let $\omega$ be any density operator supported there.
Define the support-restriction channel and the embedding channel by
\begin{equation}
 \mathcal R: A\longrightarrow A_s,
 \qquad
 \mathcal R(O)=POP+\Tr[(I_A-P)O]\,\omega,
 \qquad
 \mathcal J:A_s\longrightarrow A.
\end{equation}
Here $\mathcal J$ is the natural isometric embedding.
Let
$\Phi_s=\{p_x,\mathcal R(\rho_x)\}_{x\in\Sigma}$ and
$\rho^s_{XA_s}=(\id_X\otimes\mathcal R)(\rho_{XA})$.
Then $\mathcal J\circ\mathcal R(\rho_x)=\rho_x$ for every source state.
Given a protocol with effective channel $\cT_n$, define
\begin{equation}
 \cT_n^s\coloneqq
 \mathcal R^{\otimes n}\circ\cT_n\circ\mathcal J^{\otimes n}.
\end{equation}
The message system and shared state are unchanged.
Thus, monotonicity of fidelity gives
\begin{align}
 F_n(\Phi,\cT_n)
 &\leq\fidelity\!\left(
  (\id_{X^n}\otimes\mathcal R^{\otimes n}\circ\cT_n)
     (\rho_{X^nA^n}),
  (\id_{X^n}\otimes\mathcal R^{\otimes n})
     (\rho_{X^nA^n})
 \right)\\
 &=F_n(\Phi_s,\cT_n^s).
\end{align}
Conversely, a protocol for $\Phi_s$ can be used on $\Phi$ by
restricting to the source support before encoding and embedding into the original space after decoding, with the same fidelity and resources.
Thus the support restriction preserves the compression problem.
We henceforth use $A$ for the restricted system, on which $\rho_\Phi>0$.

\paragraph{Removal of redundancy.}
Let $\widetilde A$ have Hilbert space
$\bigoplus_{c\in\Xi}\cH_{Q_c}$, and set
\begin{equation}
 \widetilde\Phi\coloneqq
 \{p_x,\mathcal{K}_{\mathrm{off},\Phi}(\rho_x)\}_{x\in\Sigma},
 \qquad
 \widetilde\rho_{X\widetilde A}\coloneqq
 (\id_X\otimes\mathcal{K}_{\mathrm{off},\Phi})(\rho_{XA}).
\end{equation}
For a protocol $\cT_n$ on the original ensemble, define
\begin{equation}
 \widetilde\cT_n\coloneqq
 \mathcal{K}_{\mathrm{off},\Phi}^{\otimes n}
 \circ\cT_n\circ
 \mathcal{K}_{\mathrm{on},\Phi}^{\otimes n}.
\end{equation}
The KI operations are applied locally before encoding and after decoding, so this protocol uses the same message system and shared state.
Since $\mathcal{K}_{\mathrm{on},\Phi}\circ
\mathcal{K}_{\mathrm{off},\Phi}(\rho_x)=\rho_x$ for every $x$, monotonicity of fidelity yields
\begin{align}
 F_n(\Phi,\cT_n)
 &\leq\fidelity\!\left(
  (\id_{X^n}\otimes
   \mathcal{K}_{\mathrm{off},\Phi}^{\otimes n}\circ\cT_n)
       (\rho_{X^nA^n}),
  (\id_{X^n}\otimes
   \mathcal{K}_{\mathrm{off},\Phi}^{\otimes n})
       (\rho_{X^nA^n})
 \right)\\
 &=F_n(\widetilde\Phi,\widetilde\cT_n).
\end{align}
A protocol for $\widetilde\Phi$ can also be lifted to $\Phi$ by removing redundancy before encoding and reattaching it after decoding;
fidelity again cannot decrease.
The classical part and nonredundant part of the KI decomposition
are unchanged by removal of redundancy.
Hence both thresholds
$R_{\mathrm{blind},\Phi}$ and
$R_{\mathrm{blind},\Phi}^{\mathrm{EA}}$ are preserved.
We may therefore work with the redundancy-free ensemble and
relabel it as $\Phi$ on $A$.

\paragraph{Reduction to positive-definite source states.}
Fix $\lambda\in(0,1)$ independently of the block length and define
\begin{equation}
 \sigma_x\coloneqq(1-\lambda)\rho_x+\lambda\rho_\Phi,
 \qquad
 \Psi\coloneqq\{p_x,\sigma_x\}_{x\in\Sigma},
 \qquad
 \sigma_{XA}\coloneqq\sum_{x\in\Sigma}
 p_x\ketbra{x}{x}\otimes\sigma_x.
\end{equation}
The average state is unchanged, and every $\sigma_x$ is positive definite:
\begin{equation}
 \sum_xp_x\sigma_x=\rho_\Phi,
 \qquad
 \sigma_x\geq\lambda\rho_\Phi>0.
\end{equation}
Conversely, $\rho_x$ can be written as a linear combination of $\{\sigma_x\}_{x \in \Sigma}$:
\begin{equation}
 \rho_x=\frac{\sigma_x-\lambda\sum_zp_z\sigma_z}{1-\lambda}.
\end{equation}
Consequently, a channel preserves every $\rho_x$ if and only if it preserves every $\sigma_x$.
Thus $\Phi$ and $\Psi$ have the same KI decomposition.
Since their average state is also the same, their unassisted and assisted rate thresholds coincide.

Define a channel on the classical system by
\begin{equation}
 \mathcal C_\lambda(O)\coloneqq
 (1-\lambda)O+\lambda\Tr(O)\sum_{z\in\Sigma}p_z\ketbra{z}{z}.
\end{equation}
It transforms the original source into the new one:
\begin{align}
 (\mathcal C_\lambda\otimes\id_A)(\rho_{XA})
 &=(1-\lambda)\rho_{XA}
   +\lambda\sum_xp_x\ketbra{x}{x}\otimes\rho_\Phi\\
 &=\sigma_{XA}.
\end{align}
The channel $\mathcal C_\lambda^{\otimes n}$ acts only on the
reference system and commutes with the effective source channel $\cT_n$.
Therefore,
\begin{align}
 F_n(\Phi,\cT_n)
 &\leq\fidelity\!\left(
  (\mathcal C_\lambda^{\otimes n}\otimes\cT_n)
    (\rho_{X^nA^n}),
  (\mathcal C_\lambda^{\otimes n}\otimes\id_{A^n})
    (\rho_{X^nA^n})
 \right)\\
 &=\fidelity\!\left(
  (\id_{X^n}\otimes\cT_n)(\sigma_{XA}^{\otimes n}),
  \sigma_{XA}^{\otimes n}
 \right)\\
 &=F_n(\Psi,\cT_n).
\end{align}
Thus an upper bound for the fixed ensemble $\Psi$
implies the same bound for $\Phi$, without changing the message dimension or shared resource.

In the strong-converse proofs, we use these reductions and relabel the resulting ensemble as $\Phi$.
The resulting ensemble \textbf{is redundancy-free, has positive-definite source states, and satisfies $p_x>0$ for every label.}
Its KI decomposition is
\begin{equation}
 \label{eq:redundancy_free_ensemble}
 \rho_x=\bigoplus_{c\in\Xi}q_{c|x}\rho_{Q_c}^{(x,c)},
 \qquad
 q_{c|x}>0,
 \qquad
 \rho_{Q_c}^{(x,c)}>0.
\end{equation}
For this reduced ensemble, $\mathcal{K}_{\mathrm{off},\Phi}(\rho_\Phi)=\rho_\Phi$.
For each $c \in \Xi$, let $P_c$ be the projector onto $\cH_{Q_c}$; we call these the KI projectors.

\begin{remark}
    As shown in Appendix~\ref{appendix:quantum-reference},
    more general mixed-state sources with a quantum reference considered in Ref.~\cite{Khanian2020b} can be reduced to ensemble sources.
    Our strong converse results therefore apply to these
    general sources as well.
\end{remark}

\subsection{Perturbed source and KI overlap}
\label{subsec:perturbed-source}
Here, $\Phi$ denotes the redundancy-free ensemble with positive-definite source states obtained in
Sec.~\ref{subsec:reduction}.
To prove strong converse theorems, we introduce a perturbation of its average state that remains block diagonal in the original KI structure.

Define the Hermitian operator and scalar
\begin{equation}
 \label{eq:perturbation}
 S_\Phi\coloneqq
 \sum_{x\in\Sigma}p_x\rho_x\ln\rho_x-\rho_\Phi\ln\rho_\Phi,
 \qquad
 \chi_\Phi\coloneqq\Tr S_\Phi.
\end{equation}
Here $\ln$ denotes the natural logarithm, whereas all rate and entropy
expressions use base-two logarithms. In particular,
\begin{equation}
 \chi_\Phi=(\ln2)\left[
 \qentropy(\rho_\Phi)-\sum_xp_x\qentropy(\rho_x)
 \right]\geq0.
\end{equation}
Set
\begin{equation}
 Q_{\Phi,t}\coloneqq\rho_\Phi+tS_\Phi.
\end{equation}
Since $\rho_\Phi>0$, there is $t'>0$ such that
$Q_{\Phi,t}>0$ for every $t\in[0,t']$. Define
\begin{equation}
 \label{eq:perturbed-state}
 q_{\Phi,t}\coloneqq\Tr Q_{\Phi,t}=1+t\chi_\Phi,
 \qquad
 \rho_{\Phi,t}\coloneqq\frac{Q_{\Phi,t}}{q_{\Phi,t}}.
\end{equation}
The operators $S_\Phi$, $Q_{\Phi,t}$, and $\rho_{\Phi,t}$ are block diagonal in the original KI decomposition.
Define
\begin{equation}
 s_{\Phi,t,c}\coloneqq\Tr(P_c\rho_{\Phi,t}),
 \qquad
 \rho_{\Phi,t,c}\coloneqq
 \frac{P_c\rho_{\Phi,t}P_c}{s_{\Phi,t,c}}
  \in \Density(\cH_{Q_c}).
\end{equation}
Then
\begin{equation}
 \rho_{\Phi,t}=\bigoplus_{c\in\Xi}s_{\Phi,t,c}\rho_{\Phi,t,c},
 \qquad
 \sum_{c \in \Xi}s_{\Phi,t,c}=1,
 \qquad
 \Tr\rho_{\Phi,t,c}=1.
\end{equation}
Write $s_{\Phi,t}=(s_{\Phi,t,c})_{c\in\Xi} \in \Prob(\Xi)$. At $t=0$ these
quantities agree with the unperturbed averages introduced in
Sec.~\ref{subsec:blind_compression}:
\begin{equation}
 s_{\Phi,0,c}=s_{\Phi,c},
 \qquad
 \rho_{\Phi,0,c}=\overline\rho_{\Phi,c},
 \qquad
 \rho_{\Phi,0}=\rho_\Phi.
\end{equation}

Diagonalize each normalized block as
\begin{equation}
 \rho_{\Phi,t,c}=\sum_i\lambda_{\Phi,t,i|c}
 \ketbra{c,i;t}{c,i;t},
 \qquad
 \sum_i\lambda_{\Phi,t,i|c}=1,
\end{equation}
and set
\begin{equation}
 \kappa_{\Phi,t,c,i}\coloneqq
 s_{\Phi,t,c}\lambda_{\Phi,t,i|c}.
\end{equation}
These are the eigenvalues of $\rho_{\Phi,t}$. For a fixed ensemble
and perturbation parameter, we may abbreviate
\begin{equation}
 s_c=s_{\Phi,t,c},\quad
 \rho_c=\rho_{\Phi,t,c},\quad
 \lambda_{i|c}=\lambda_{\Phi,t,i|c},\quad
 \kappa_{c,i}=\kappa_{\Phi,t,c,i},\quad
 \ket{c,i}=\ket{c,i;t}.
\end{equation}

Let $R$ be a copy of $A$. With complex conjugation taken in the
fixed vectorization basis, the canonical purification is
\begin{equation}
 \vk{\sqrt{\rho_{\Phi,t}}}
 =\sum_{c,i}\sqrt{\kappa_{c,i}}\,
 \ket{c,i}_A\otimes\overline{\ket{c,i}}_R,
 \qquad
 \psi_{\Phi,t}\coloneqq
 \vkb{\sqrt{\rho_{\Phi,t}}}{\sqrt{\rho_{\Phi,t}}}.
\end{equation}
Its marginals are $\rho_{\Phi,t}$ on $A$ and
$\rho_{\Phi,t}^{\Trans}$ on $R$. Define the normalized block vectors
\begin{equation}
 \label{eq:ki-projector}
 \ket{\varphi_{\Phi,t,c}}\coloneqq
 \frac{\vk{P_c\sqrt{\rho_{\Phi,t}}}}{\sqrt{s_{\Phi,t,c}}}
 =\sum_i\sqrt{\lambda_{i|c}}\,
 \ket{c,i}_A\otimes\overline{\ket{c,i}}_R.
\end{equation}
They are orthonormal because distinct KI blocks have orthogonal
supports. Hence
\begin{equation}
 \Pi_{\Phi,t}\coloneqq
 \sum_{c\in\Xi}\ketbra{\varphi_{\Phi,t,c}}{\varphi_{\Phi,t,c}}
\end{equation}
is a projector. For any quantum channel $\cT_\ell:A^\ell\to A^\ell$,
define
\begin{equation}
 \label{eq:overlap}
 \begin{aligned}
  \Omega_{\Phi,t}(\cT_\ell)
  &\coloneqq(\cT_\ell\otimes\id_{R^\ell})
     (\psi_{\Phi,t}^{\otimes\ell}),\\
  g_\ell(\Phi,t,\cT_\ell)
  &\coloneqq\Tr\!\left[
     \Omega_{\Phi,t}(\cT_\ell)\Pi_{\Phi,t}^{\otimes\ell}
     \right].
 \end{aligned}
\end{equation}
We call $g_\ell(\Phi,t,\cT_\ell)\in[0,1]$ the \emph{KI overlap}: 
it quantifies how well the channel preserves the KI structure of the source by measuring the weight of the output state in the subspace spanned by the normalized KI block purifications.

The following lemma relates the compression fidelity to the KI overlap.

\begin{lemma}[Compression fidelity and overlap]
\label{lem:fidelity-comparison}
Let $\Phi$ be a redundancy-free ensemble with positive-definite source states.
There exist $t_0\in(0,t']$ and constants $C,b>0$, depending only on $\Phi$, such that for every positive integer $n$,
every channel $\cT_n$ on $A^n$, and every $t\in(0,t_0]$,
\begin{equation}
\label{eq:fidelity-comparison}
F_n(\Phi,\cT_n)
\leq e^{Ct^2n}
\Tr\!\left[
\Omega_{\Phi,t}(\cT_n)
\left(\Pi_{\Phi,t}+e^{-bt}(I-\Pi_{\Phi,t})\right)^{\otimes n}
\right].
\end{equation}
\end{lemma}
\begin{proof}
    See Appendix~\ref{appendix:proof_lem_comparison}.
\end{proof}

\section{Strong converse for unassisted blind compression}
\label{sec:main_results}
In this section, we prove an exponential strong converse for unassisted blind compression.

\begin{tcolorbox}[
  colback=blue!3,
  colframe=blue!60!black,
  boxrule=0.5pt,
  arc=2mm,
  breakable
]
\begin{theorem}[Exponential strong converse]\label{thm:main}
Let $\Phi$ be a quantum ensemble on system $A$.
For every $0\leq r <  R_{\mathrm{blind},\Phi}$, there exist a constant $\alpha > 0$ and a positive integer $n_0$, depending only on $\Phi$ and $r$, such that for every $n \geq n_0$ and every unassisted code $(\cE_n,\cD_n)$ with
\begin{equation}
 \cE_n:A^{\otimes n}\longrightarrow M_n,
 \qquad
 \cD_n:M_n\longrightarrow A^{\otimes n},
 \qquad \log_2\dim M_n\le nr,
\end{equation}
we have
\begin{equation}\label{eq:conclusion}
 F_n(\Phi,\cD_n \circ \cE_n) \leq 2^{-\alpha n}.
\end{equation}
\end{theorem}
\end{tcolorbox}

By the source reductions in Sec.~\ref{subsec:reduction},
it suffices to consider a redundancy-free ensemble whose
states are positive definite.
The proof proceeds as follows.
In Sec.~\ref{subsec:bound_KIoverlap}, we bound the KI overlap in terms of the message-system dimension and the entropy of the perturbed average state (Lemma~\ref{lem:memory-bound}).
Then, in Sec.~\ref{subsec:proof_strong_converse}, we combine the fidelity bound in Lemma~\ref{lem:fidelity-comparison} with the KI-overlap bound in Lemma~\ref{lem:memory-bound} to prove Theorem~\ref{thm:main}.

\subsection{Bound on KI overlap}
In this subsection, we bound the KI overlap in terms of the message-system dimension.
\label{subsec:bound_KIoverlap}
\begin{lemma}[Dimension bound on the KI overlap]
\label{lem:memory-bound}
Let $\Phi$ be a redundancy-free ensemble with positive-definite source states.
There exists $t_1\in(0,t']$, depending only on $\Phi$, with
the following property.
For every $\delta>0$, there exists $v>0$, depending only on $\Phi$ and $\delta$, such that for every $t\in (0,t_1]$, every positive integer $\ell$, and every channel $\cT_\ell=\cD_\ell\circ\cE_\ell$ with
$\cE_\ell: A^\ell \to M$ and $\cD_\ell: M \to A^\ell$, it holds that
\begin{equation}
\label{eq:memory-bound}
g_\ell(\Phi,t,\cT_\ell)
\leq(\dim M)\,2^{-\ell(S(\rho_{\Phi,t})-\delta)}
+e^{-v\ell}.
\end{equation}
\end{lemma}
\begin{proof}
Diagonalize $\rho_{\Phi,t}^{\otimes\ell}$ within the original KI decomposition:
\begin{equation}
\rho_{\Phi,t}^{\otimes\ell}
=\bigoplus_{\boldsymbol c}s_{\boldsymbol c}
\sum_{\boldsymbol i}\lambda_{\boldsymbol i | \boldsymbol c}
|\boldsymbol c,\boldsymbol i\rangle\langle\boldsymbol c,\boldsymbol i|,
\end{equation}
where
\begin{equation}
    \sum_{\boldsymbol c}s_{\boldsymbol c} = 1, \qquad 
    \sum_{\boldsymbol i}\lambda_{\boldsymbol i | \boldsymbol c} = 1 \quad \forall \boldsymbol c.
\end{equation}
Here, we write $\boldsymbol c \coloneqq c_1c_2\ldots c_\ell$, where $c_j \in \Xi$ for all $1\leq j \leq \ell$, 
and the eigenvalues of $\rho_{\Phi,t}^{\otimes\ell}$ are
\begin{equation}
\kappa_{\boldsymbol c,\boldsymbol i}
\coloneqq s_{\boldsymbol c}\lambda_{\boldsymbol i|\boldsymbol c}
=\prod_{j=1}^{\ell}s_{c_j}\lambda_{i_j|c_j}
=\prod_{j=1}^{\ell}\kappa_{c_j,i_j}.
\end{equation}

For any choice of Kraus operators $\{T_{\ell,a}\}_a$ of $\cT_\ell$, we have 
\begin{equation}
g_\ell(\Phi,t,\cT_\ell)
=\sum_{\boldsymbol c,a}s_{\boldsymbol c}
\left|\sum_{\boldsymbol i}\lambda_{\boldsymbol i | \boldsymbol c}
\langle\boldsymbol c, \boldsymbol i|T_{\ell,a}|\boldsymbol c, \boldsymbol i\rangle\right|^2.
\end{equation}
For every $\boldsymbol c$ and $a$, the Cauchy-Schwarz inequality gives
\begin{equation}
    \left|\sum_{\boldsymbol i}\lambda_{\boldsymbol i|\boldsymbol c}\braket{\boldsymbol c, \boldsymbol i|T_{\ell,a}|\boldsymbol c, \boldsymbol i}\right|^2
    \leq \left(\sum_{\boldsymbol i} \lambda_{\boldsymbol i|\boldsymbol c}\right)\left(\sum_{\boldsymbol i} \lambda_{\boldsymbol i|\boldsymbol c}\left|\braket{\boldsymbol c, \boldsymbol i|T_{\ell,a}|\boldsymbol c, \boldsymbol i}\right|^2\right)
    =\sum_{\boldsymbol i} \lambda_{\boldsymbol i|\boldsymbol c}\left|\braket{\boldsymbol c, \boldsymbol i|T_{\ell,a}|\boldsymbol c, \boldsymbol i}\right|^2,
\end{equation}
where the equality follows from $\sum_{\boldsymbol i} \lambda_{\boldsymbol i|\boldsymbol c} = 1$.
Hence,
\begin{equation}\label{eq:g_ell_eval_1}
    g_\ell(\Phi,t,\cT_\ell)
    \leq \sum_{\boldsymbol c,a} s_{\boldsymbol c} \sum_{\boldsymbol i} \lambda_{\boldsymbol i|\boldsymbol c} \left|\braket{\boldsymbol c, \boldsymbol i|T_{\ell,a}|\boldsymbol c, \boldsymbol i}\right|^2
    = \sum_{\boldsymbol c}\sum_{\boldsymbol i} s_{\boldsymbol c}\lambda_{\boldsymbol i|\boldsymbol c} \left(\sum_a \left|\braket{\boldsymbol c, \boldsymbol i|T_{\ell,a}|\boldsymbol c, \boldsymbol i}\right|^2\right).
\end{equation}
Define
\begin{align}
    w_{\boldsymbol c, \boldsymbol i}
    =\langle\boldsymbol c, \boldsymbol i|\cT_\ell(|\boldsymbol c, \boldsymbol i\rangle\langle\boldsymbol c, \boldsymbol i|)|\boldsymbol c, \boldsymbol i\rangle
    =\sum_a|\braket{\boldsymbol c, \boldsymbol i|T_{\ell,a}|\boldsymbol c, \boldsymbol i}|^2,
\end{align}
and we have
\begin{equation}
\label{eq:g_ell_upper_bound}
g_\ell\leq\sum_{\boldsymbol c, \boldsymbol i}\kappa_{\boldsymbol c, \boldsymbol i}
w_{\boldsymbol c, \boldsymbol i}.
\end{equation}
By definition, $0\leq w_{\boldsymbol c, \boldsymbol i}\leq 1$, and it holds that
\begin{equation}\label{eq:memory}
 \begin{aligned}
 \sum_{\boldsymbol c}\sum_{\boldsymbol i} w_{\boldsymbol c, \boldsymbol i}
 &= \sum_{\boldsymbol c}\sum_{\boldsymbol i} \Tr \left[
   \ketbra{\boldsymbol c, \boldsymbol i}{\boldsymbol c, \boldsymbol i}\cT_\ell(\ketbra{\boldsymbol c, \boldsymbol i}{\boldsymbol c, \boldsymbol i})
 \right]\\
 &=\sum_{\boldsymbol c}\sum_{\boldsymbol i} \Tr\left[
   \cE_\ell(\ketbra{\boldsymbol c, \boldsymbol i}{\boldsymbol c, \boldsymbol i})\cD_\ell^\dagger(\ketbra{\boldsymbol c, \boldsymbol i}{\boldsymbol c, \boldsymbol i})
 \right]\\
 &\leq\sum_{\boldsymbol c}\sum_{\boldsymbol i}\Tr\cD_\ell^\dagger(\ketbra{\boldsymbol c, \boldsymbol i}{\boldsymbol c, \boldsymbol i})\\
 &=\Tr I_{M} = \dim M.
 \end{aligned}
\end{equation}

Let
\begin{equation}
G_\ell(t)=\{(\boldsymbol c, \boldsymbol i):\kappa_{\boldsymbol c, \boldsymbol i}\leq2^{-\ell(S(\rho_{\Phi,t})-\delta)}\},
\qquad
\beta_\ell(t)=\sum_{(\boldsymbol c, \boldsymbol i)\notin G_\ell(t)}\kappa_{\boldsymbol c, \boldsymbol i}.
\end{equation}
Splitting the sum in~\eqref{eq:g_ell_upper_bound} over $G_\ell(t)$ and its complement yields
\begin{equation}
\begin{aligned}
    g_\ell(\Phi,t,\cT_\ell)
    &\leq \sum_{(\boldsymbol c, \boldsymbol i) \in G_\ell(t)} \kappa_{\boldsymbol c, \boldsymbol i} w_{\boldsymbol c, \boldsymbol i} + \sum_{(\boldsymbol c, \boldsymbol i) \notin G_\ell(t)} \kappa_{\boldsymbol c, \boldsymbol i} w_{\boldsymbol c, \boldsymbol i} \\
    & \leq 2^{-\ell(S(\rho_{\Phi,t})-\delta)}\sum_{(\boldsymbol c, \boldsymbol i) \in G_\ell(t)} w_{\boldsymbol c, \boldsymbol i} + \sum_{(\boldsymbol c, \boldsymbol i) \notin G_\ell(t)} \kappa_{\boldsymbol c, \boldsymbol i} \\
    &\leq(\dim M)\,2^{-\ell(S(\rho_{\Phi,t})-\delta)}
+\beta_\ell(t).
\end{aligned}
\end{equation}

It remains to bound $\beta_\ell(t)$.
First, by Lemma~\ref{lem:uniform-positive},
there exist $t_1\in(0,t']$ and $\widetilde \kappa \in (0,1/2]$, depending only on $\Phi$, such that
\begin{equation}
\kappa_{c,i}\geq\widetilde\kappa
\end{equation}
for every $(c,i)$ and $t \in [0,t_1]$.
Put $L \coloneqq -\log_2\widetilde\kappa>0$.
Now, let $(\mathsf{C}_j,\mathsf{I}_j)$, $j=1,\ldots,\ell$, be independent and identically distributed random pairs with
\begin{equation}
\Pr[\mathsf{C}_j=c,\mathsf{I}_j=i]=\kappa_{c,i},
\end{equation}
and set
\begin{equation}
    \mathsf{X}_j=-\log_2\kappa_{\mathsf{C}_j,\mathsf{I}_j}.
\end{equation}
Since
\begin{equation}
0\leq \mathsf{X}_j\leq L,
\qquad
\mathbb E[\mathsf{X}_j]= \sum_{c,i} \kappa_{c,i}(-\log_2 \kappa_{c,i})=S(\rho_{\Phi,t}),
\end{equation}
Hoeffding's inequality~\cite[Sec.~12.2]{GrimmettStirzaker2001} gives
\begin{equation}
\begin{aligned}
\beta_\ell(t)
&=\Pr\!\left[
\prod_{j=1}^{\ell}\kappa_{\mathsf{C}_j,\mathsf{I}_j}
>2^{-\ell(S(\rho_{\Phi,t})-\delta)}\right]\\
&=\Pr\!\left[
\sum_{j=1}^{\ell}\mathsf{X}_j
<\ell(S(\rho_{\Phi,t})-\delta)\right]\\
&\leq\Pr\!\left[
\sum_{j=1}^{\ell}
\bigl(\mathsf{X}_j-S(\rho_{\Phi,t})\bigr)
\leq-\ell\delta\right]\\
&\leq\exp\!\left(-\frac{2\delta^2}{L^2}\ell\right).
\end{aligned}
\end{equation}
Taking $v=2\delta^2/L^2$ proves the claim.

\end{proof}

\subsection{Proof of strong converse}
\label{subsec:proof_strong_converse}
We now combine Lemma~\ref{lem:fidelity-comparison} and Lemma~\ref{lem:memory-bound} to prove Theorem~\ref{thm:main}.

\begin{proof}[Proof of Theorem~\ref{thm:main}]
By the reductions in Sec.~\ref{subsec:reduction}, it suffices to consider a redundancy-free ensemble with positive-definite source states.

Let $t_0,C,b>0$ be the constants supplied by Lemma~\ref{lem:fidelity-comparison}.
Write $\Pi=\Pi_{\Phi,t}$, $\Pi^\perp=I-\Pi$, and set
\begin{equation}
\Theta_t \coloneqq \Pi+e^{-bt}\Pi^\perp,
\qquad
k \coloneqq \lfloor\eta n\rfloor,
\end{equation}
where $0<\eta<1/2$ and $t>0$ will be fixed below.
For each subset $S\subseteq[n]$, let
\begin{equation}
P_S \coloneqq \bigotimes_{i=1}^n
\begin{cases}
\Pi^\perp,&i\in S,\\
\Pi,&i\notin S.
\end{cases}
\qquad
P_{\leq k} \coloneqq \sum_{\substack{S\subseteq[n]\\|S|\leq k}}P_S.
\end{equation}
Since the projectors $P_S$ are mutually orthogonal, we have
\begin{equation}
\Theta_t^{\otimes n}
=\sum_{S\subseteq[n]}e^{-bt|S|}P_S
\leq P_{\leq k}+e^{-bt(k+1)}I
\leq P_{\leq k}+e^{-bt\eta n}I.
\end{equation}
Moreover,
\begin{equation}
P_{\leq k}\leq \sum_{\substack{S\subseteq[n]\\|S|\leq k}}I_{A_SR_S}\otimes\Pi^{\otimes S^c}.
\end{equation}
Here and below, tensor factors are reordered and placed at the indicated positions.
Define
\begin{equation}
g_S \coloneqq \Tr\!\left[
\Omega_{\Phi,t}(\cT_n)
(I_{A_SR_S}\otimes\Pi^{\otimes S^c})\right].
\end{equation}
By Lemma~\ref{lem:fidelity-comparison}, we have 
\begin{equation}
\begin{aligned}
    \label{eq:fidelity_bound}
    F_n
    &\leq e^{Ct^2n}\Tr\!\left[\Omega_{\Phi,t}(\cT_n)\Theta_t^{\otimes n}\right]\\
    &\leq e^{Ct^2n}\left(\sum_{\substack{S\subseteq[n]\\|S|\leq k}}\Tr\!\left[\Omega_{\Phi,t}(\cT_n)\left(I_{A_SR_S}\otimes\Pi^{\otimes S^c}\right)\right] + e^{-bt\eta n}\Tr\!\left[\Omega_{\Phi,t}(\cT_n)\right]\right)\\
    &\leq e^{Ct^2n}\left(\sum_{\substack{S\subseteq[n]\\|S|\leq k}}g_S+e^{-bt\eta n}\right)
\end{aligned}
\end{equation}
for all $t \in (0,t_0]$.

Fix $S \subseteq [n]$ with $|S|\leq k$ and set $\ell=n-|S|$.
Define a channel $\cT_{n,S}: A_{S^c} \to A_{S^c}$ by
\begin{equation}
\cT_{n,S}(X_{A_{S^c}}) \coloneqq \Tr_{A_S}
\cT_n(X_{A_{S^c}}\otimes\rho_{\Phi,t}^{\otimes S}).
\end{equation}
Note that by definition
\begin{equation}
\Tr_{A_SR_S}\Omega_{\Phi,t}(\cT_n)
=(\cT_{n,S}\otimes\id_{R_{S^c}})
(\psi_{\Phi,t}^{\otimes\ell}),
\end{equation}
and thus 
\begin{equation}
    g_S=g_\ell(\Phi,t,\cT_{n,S}).
\end{equation}
If $\cT_n=\cD_n\circ\cE_n$, then we can write $\cT_{n,S} = \cD_{n,S} \circ \cE_{n,S}$, where
\begin{equation}
\cE_{n,S}(X)
=\cE_n(X\otimes\rho_{\Phi,t}^{\otimes S}),
\qquad
\cD_{n,S}(Y)=\Tr_{A_S}\cD_n(Y), 
\end{equation}
in particular, $\cE_{n,S}:A_{S^c} \to M_n$ and $\cD_{n,S}: M_n \to A_{S^c}$. 
Now, we can apply Lemma~\ref{lem:memory-bound} to $g_S$. 
Letting $t_1 > 0$ be a constant supplied by Lemma~\ref{lem:memory-bound}, 
for every $\delta > 0$, there exists $v > 0$ such that
\begin{equation}
g_S\leq (\dim M_n)\,2^{-\ell(S(\rho_{\Phi,t})-\delta)} +e^{-v\ell}
\end{equation}
for all $t \in (0,t_1]$.
Here, if $S(\rho_{\Phi,t})-\delta>0$, $\ell\geq n(1-\eta)$ implies that
\begin{equation}
\label{eq:g_s_bound}
g_S\leq
2^{nr-n(1-\eta)(S(\rho_{\Phi,t})-\delta)}
+e^{-vn(1-\eta)}.
\end{equation}

The number of subsets in the sum is bounded as follows:
\begin{equation}
\label{eq:counting_estimate}
\sum_{\substack{S\subseteq[n]\\|S|\leq k}} 
= \sum_{s=0}^{\lfloor\eta n\rfloor}\binom ns
\leq2^{nh_2(\eta)}.
\end{equation}
Indeed, for $s\leq\eta n$,
\begin{equation}
\eta^s(1-\eta)^{n-s}
\geq\eta^{\eta n}(1-\eta)^{(1-\eta)n}
=2^{-nh_2(\eta)},
\end{equation}
and thus
\begin{equation}
    1 \geq \sum_{s=0}^{\lfloor\eta n\rfloor}\binom ns\eta^s(1-\eta)^{n-s} \geq 2^{-nh_2(\eta)}\sum_{s=0}^{\lfloor\eta n\rfloor}\binom ns .
\end{equation}
Combining \eqref{eq:fidelity_bound}, \eqref{eq:g_s_bound}, and \eqref{eq:counting_estimate}, if $S(\rho_{\Phi,t})-\delta>0$, we have 
\begin{equation}
\label{eq:unassisted-central}
F_n\leq e^{Ct^2n}
\left\{
2^{nh_2(\eta)}
\left[
2^{nr-n(1-\eta)(S(\rho_{\Phi,t})-\delta)}
+e^{-vn(1-\eta)}
\right]+e^{-bt\eta n}
\right\}.
\end{equation}

Now, set
\begin{equation}
\Delta=S(\rho_\Phi)-r>0,
\qquad
\delta=\Delta/8.
\end{equation}
By continuity of the von Neumann entropy, there exists a sufficiently small $t_2>0$ such that
\begin{equation}
S(\rho_{\Phi,t})\geq S(\rho_\Phi)-\delta
\end{equation}
for all $0\leq t\leq t_2$.
In particular, we have $S(\rho_{\Phi,t})-\delta \geq S(\rho_\Phi)-2\delta >0$ for this choice. 
Fix the constant $v>0$ supplied by Lemma~\ref{lem:memory-bound} for this choice of $\delta$.
Choose $\eta\in(0,1/2)$ sufficiently small so that
\begin{equation}
\begin{aligned}
A_1&=(\ln2)\left[
(1-\eta)(S(\rho_\Phi)-2\delta)-r-h_2(\eta)
\right]>0,\\
A_2&=v(1-\eta)-(\ln2)h_2(\eta)>0.
\end{aligned}
\end{equation}
This is possible because their limits as $\eta \to 0^+$ are
\begin{align*}
    \lim_{\eta \to 0^+} A_1 &= (\ln2)\left[S(\rho_\Phi)-2\delta-r \right] \geq (\ln2)\frac{3}{4}\Delta > 0 \\
    \lim_{\eta \to 0^+} A_2 &= v > 0.
\end{align*}
Finally, choose $t$ with
\begin{equation}
0<t\leq\min\{t_0,t_1,t_2\}
\quad\text{such that}\quad
Ct^2<\frac12\min\{A_1,A_2,bt\eta\}.
\end{equation}
By taking
\begin{equation}
\widetilde\alpha
=\min\{A_1-Ct^2,A_2-Ct^2,bt\eta-Ct^2\}>0,
\end{equation}
we obtain
\begin{equation}
F_n\leq3e^{-\widetilde\alpha n}.
\end{equation}
Take
\begin{equation}
\alpha=\frac{\widetilde\alpha}{2\ln2},
\qquad
n_0=\max\left\{1,
\left\lceil\frac{2\ln3}{\widetilde\alpha}\right\rceil\right\}.
\end{equation}
Then, for every $n\geq n_0$,
\begin{equation}
F_n\leq3e^{-\widetilde\alpha n}
\leq e^{-\widetilde\alpha n/2}=2^{-\alpha n},
\end{equation}
which proves the claim.
\end{proof}

\section{Strong converse for blind compression with entanglement assistance}
\label{sec:assisted}
This section treats entanglement-assisted compression in two stages.
In Sec.~\ref{subsec:unlimited_assistance}, we first prove an exponential strong converse for blind compression with unlimited entanglement assistance at the optimal assisted quantum communication rate. 
Then, in Sec.~\ref{subsec:strong_converse_rate_region}, we combine this assisted converse bound with the unassisted converse in Sec.~\ref{sec:main_results} to obtain the complete strong-converse qubit-ebit region. 

\subsection{Unlimited entanglement assistance}
\label{subsec:unlimited_assistance}
Here, we establish an exponential strong converse for compression with unlimited entanglement assistance.

\begin{tcolorbox}[
  colback=blue!3,
  colframe=blue!60!black,
  boxrule=0.5pt,
  arc=2mm,
  breakable
]
\begin{theorem}[Entanglement-assisted exponential strong converse]
\label{thm:assisted}
Let $\Phi$ be a quantum ensemble on system $A$.
For every
\begin{equation}
0\leq r<R_{\mathrm{blind},\Phi}^{\mathrm{EA}},
\end{equation}
there exist $\alpha>0$ and a positive integer $n_0$, depending only on
$\Phi$ and $r$, such that the following holds.
For every $n\geq n_0$,
every finite-dimensional source-independent shared state
$\zeta_n \in \Density(\cH_{A_0} \otimes \cH_{B_0})$, and every pair of encoding and decoding channels
\begin{equation}
\cE_n:A^{\otimes n}A_0\longrightarrow M_n,
\qquad
\cD_n:M_nB_0\longrightarrow A^{\otimes n},
\qquad
\log_2\dim M_n\leq nr,
\end{equation}
the effective source channel
\begin{equation}
\cT_n(X)
\coloneqq\cD_n\!\left[
(\cE_n\otimes\id_{B_0})(X\otimes\zeta_n)
\right]
\end{equation}
satisfies
\begin{equation}
F_n(\Phi,\cT_n)\leq2^{-\alpha n}.
\end{equation}
\end{theorem}
\end{tcolorbox}

Note that this theorem proves a strong converse for all finite-dimensional source-independent shared states $\zeta_n$ 
while the achievability result in Theorem~\ref{thm:optimal_blind_compression_rate} used a maximally entangled state as a shared resource. 

By the source reductions in Sec.~\ref{subsec:reduction},
it suffices to consider a redundancy-free ensemble whose
states are positive definite.
Lemma~\ref{lem:fidelity-comparison} continues to apply because the effective channel $\cT_n$ is completely positive and trace preserving.
The additional difficulty is that the decoder now receives both the message ($M_n$) and a share ($B_0$) of an arbitrarily large resource, so a bound using the full input dimension of the decoder, as in Lemma~\ref{lem:memory-bound}, would not suffice.
We instead prove an operator domination bound that depends only on the message-system dimension (Lemma~\ref{lem:message-domination}).
We then convert this bound into an estimate on the KI overlap (Lemma~\ref{lem:assisted-overlap}) and use that estimate to prove the strong converse.

Recall the perturbed average and its normalized block states from
Sec.~\ref{subsec:perturbed-source}:
\begin{equation}
\rho_{\Phi,t}
=\bigoplus_{c\in\Xi}s_{\Phi,t,c}\rho_{\Phi,t,c},
\qquad
s_{\Phi,t,c}=\Tr(P_c\rho_{\Phi,t}).
\end{equation}
Define
\begin{equation}
\label{eq:assisted-entropy}
K_{\Phi,t}
\coloneqq H(s_{\Phi,t})
+2\sum_{c\in\Xi}s_{\Phi,t,c}S(\rho_{\Phi,t,c})
=2S(\rho_{\Phi,t})-H(s_{\Phi,t}).
\end{equation}
By definition and continuity,
\begin{equation}
K_{\Phi,0}=2R_{\mathrm{blind},\Phi}^{\mathrm{EA}},
\qquad
\lim_{t\to 0^+}K_{\Phi,t}=K_{\Phi,0}.
\end{equation}

We first prove an operator domination bound controlled by the message-system dimension.
\begin{lemma}[Message-dimension domination]
\label{lem:message-domination}
Let $\xi_{AR}$ be a density operator, and let $\zeta_{A_0B_0}$ be a
finite-dimensional shared state independent of $AR$. Consider channels
\begin{equation}
\cE:AA_0\longrightarrow M,
\qquad
\cD:MB_0\longrightarrow B.
\end{equation}
Define the output state by
\begin{equation}
\Omega_{BR}
\coloneqq(\cD\otimes\id_R)
\!\left[(\cE\otimes\id_{B_0R})
(\xi_{AR}\otimes\zeta_{A_0B_0})\right].
\end{equation}
There is a density operator $\sigma_B$ such that
\begin{equation}
\label{eq:message-domination}
\Omega_{BR}\leq(\dim M)^2\sigma_B\otimes\xi_R.
\end{equation}
\end{lemma}
\begin{proof}
Put $m\coloneqq\dim M$ and
$\pi_M\coloneqq I_M/m$. The state immediately after encoding is
\begin{equation}
\vartheta_{MB_0R}
\coloneqq(\cE\otimes\id_{B_0R})
(\xi_{AR}\otimes\zeta_{A_0B_0}).
\end{equation}
Since the initial resource is independent of the source and the encoder is trace preserving, the marginal on the joint system $B_0R$ is
\begin{equation}
\vartheta_{B_0R}=\zeta_{B_0}\otimes\xi_R.
\end{equation}
By Lemma~\ref{lem:pinching}, we have
\begin{equation}
\vartheta_{MB_0R}
\leq m I_M\otimes\zeta_{B_0}\otimes\xi_R
=m^2\pi_M\otimes\zeta_{B_0}\otimes\xi_R.
\end{equation}

Define $\sigma_B\coloneqq\cD(\pi_M\otimes\zeta_{B_0})$, which is a
density operator because $\cD$ is a quantum channel. Complete positivity
of $\cD$ implies
\begin{equation}
\Omega_{BR}
=(\cD\otimes\id_R)(\vartheta_{MB_0R})
\leq m^2\sigma_B\otimes\xi_R,
\end{equation}
as required.
\end{proof}
We next apply this domination bound to the canonical purification to obtain a bound on the KI overlap. 
The following lemma gives an upper bound analogous to that of Lemma~\ref{lem:memory-bound}.

\begin{lemma}[Assisted KI-overlap bound]
\label{lem:assisted-overlap}
Let $\Phi$ be redundancy-free with positive-definite source states.
There exists $\widetilde t_1\in(0,t']$, depending only on $\Phi$, with
the following property. For every $\delta>0$, there is
$\widetilde v>0$, depending only on $\Phi$ and $\delta$, such that for
every $t\in(0,\widetilde t_1]$, every positive integer $\ell$, and
every effective channel
\begin{equation}
\cT_\ell(X)
=\cD_\ell\!\left[(\cE_\ell\otimes\id_{B_0})
(X\otimes\zeta_{\ell}^{A_0B_0})\right],
\end{equation}
with 
\begin{equation}
\cE_\ell:A^\ell A_0\longrightarrow M,
\qquad
\cD_\ell:MB_0\longrightarrow A^\ell
\end{equation}
and any finite-dimensional source-independent shared state $\zeta_{\ell}^{A_0B_0}$, 
it holds that 
\begin{equation}
\label{eq:assisted-overlap}
g_\ell(\Phi,t,\cT_\ell)
\leq2(\dim M)^2\,
2^{-\ell(K_{\Phi,t}-\delta)}+2e^{-\widetilde v\ell}.
\end{equation}
\end{lemma}
\begin{proof}
For the proof, we use the following shorthand:
\begin{equation}
\rho\coloneqq\rho_{\Phi,t},\quad
s_c\coloneqq s_{\Phi,t,c},\quad
\rho_c\coloneqq\rho_{\Phi,t,c},\quad
\lambda_{i|c}\coloneqq\lambda_{\Phi,t,i|c},\quad
\Pi\coloneqq\Pi_{\Phi,t},\quad
m\coloneqq\dim M.
\end{equation}
We use the block eigenvectors of Sec.~\ref{subsec:perturbed-source},
with their dependence on $\Phi,t$ also suppressed, so that
\begin{equation}
\rho_c=\sum_i\lambda_{i|c}|c,i\rangle\langle c,i|,
\qquad
|\varphi_c\rangle
=\sum_i\sqrt{\lambda_{i|c}}\,
|c,i\rangle_A\otimes\overline{|c,i\rangle}_R.
\end{equation}
Here $|\varphi_c\rangle=|\varphi_{\Phi,t,c}\rangle$, and complex
conjugation is taken in the fixed vectorization basis. In particular,
the reference marginal of the canonical purification is
$\rho^{\Trans}$. Direct contraction in a single block gives
\begin{equation}
\begin{aligned}
\Tr_R\!\left[(I_A\otimes\rho^{\Trans})
|\varphi_c\rangle\langle\varphi_c|\right]
&=\sum_i s_c\lambda_{i|c}^2|c,i\rangle\langle c,i|\\
&=s_c\rho_c^2.
\end{aligned}
\end{equation}
Summing over the blocks yields
\begin{equation}
\Tr_R[(I_A\otimes\rho^{\Trans})\Pi]
=\bigoplus_c s_c\rho_c^2.
\end{equation}

For product indices, put
\begin{equation}
\begin{aligned}
|\boldsymbol c,\boldsymbol i\rangle
\coloneqq\bigotimes_{j=1}^{\ell}|c_j,i_j\rangle,
\qquad
s_{\boldsymbol c}\coloneqq\prod_{j=1}^{\ell}s_{c_j},
\qquad
\lambda_{\boldsymbol i|\boldsymbol c}
\coloneqq\prod_{j=1}^{\ell}\lambda_{i_j|c_j}.
\end{aligned}
\end{equation}
Set
\begin{equation}
a_\ell(t)\coloneqq2^{-\ell(K_{\Phi,t}-\delta)},
\qquad
\widetilde G_\ell(t)
\coloneqq\{(\boldsymbol c,\boldsymbol i):
s_{\boldsymbol c}\lambda_{\boldsymbol i|\boldsymbol c}^{2}
\leq a_\ell(t)\}.
\end{equation}
Define the projector onto the space corresponding to $\widetilde G_\ell(t)$ by
\begin{equation}
P_{\widetilde G_\ell(t)}
\coloneqq\sum_{(\boldsymbol c,\boldsymbol i)\in\widetilde G_\ell(t)}
\overline{\ket{\boldsymbol c,\boldsymbol i}}\,
\overline{\bra{\boldsymbol c,\boldsymbol i}}.
\end{equation}
Write
\begin{equation}
P_\ell\coloneqq I_{A^\ell}\otimes
P_{\widetilde G_\ell(t)},
\qquad
\Pi_\ell\coloneqq\Pi^{\otimes\ell},
\qquad
\Omega\coloneqq\Omega_{\Phi,t}(\cT_\ell).
\end{equation}

Applying Lemma~\ref{lem:message-domination} to $\Omega$ gives a state $\sigma$ on $A^\ell$ such that
\begin{equation}
\Omega\leq m^2\sigma\otimes(\rho^{\Trans})^{\otimes\ell}
\end{equation}
since $\Omega_{R^\ell}$ is equal to the reference marginal of the $\ell$-fold canonical purification, which is $(\rho^\Trans)^{\otimes \ell}$. 
Then, direct computation gives
\begin{equation}
\begin{aligned}
B&\coloneqq\Tr_{R^\ell}\!\left[
\bigl(I_{A^\ell}\otimes(\rho^{\Trans})^{\otimes\ell}\bigr)
P_\ell\Pi_\ell P_\ell\right]\\
&=\sum_{(\boldsymbol c,\boldsymbol i)\in\widetilde G_\ell(t)}
s_{\boldsymbol c}\lambda_{\boldsymbol i|\boldsymbol c}^{2}
|\boldsymbol c,\boldsymbol i\rangle
\langle\boldsymbol c,\boldsymbol i|\\
&\leq a_\ell(t) I_{A^\ell}.
\end{aligned}
\end{equation}
Consequently,
\begin{equation}
\Tr(\Omega P_\ell\Pi_\ell P_\ell)
\leq m^2\Tr(\sigma B)\leq m^2a_\ell(t).
\end{equation}
On the other hand, 
\begin{equation}
\widetilde \beta_\ell(t)
\coloneqq\Tr[\Omega(I-P_\ell)]
=\sum_{(\boldsymbol c,\boldsymbol i)\notin\widetilde G_\ell(t)}
s_{\boldsymbol c}\lambda_{\boldsymbol i|\boldsymbol c}.
\end{equation}
By the Hilbert--Schmidt triangle inequality, 
we separate the two parts:
\begin{equation}
\begin{aligned}
\sqrt{g_\ell(\Phi,t,\cT_\ell)}
&=\|\Pi_\ell\Omega^{1/2}\|_2\\
&\leq\|\Pi_\ell P_\ell\Omega^{1/2}\|_2
+\|\Pi_\ell(I-P_\ell)\Omega^{1/2}\|_2\\
&\leq m\sqrt{a_\ell(t)}+\sqrt{\widetilde \beta_\ell(t)}.
\end{aligned}
\end{equation}
Squaring and using $(u+v)^2\leq2u^2+2v^2$ gives
\begin{equation}
g_\ell(\Phi,t,\cT_\ell)
\leq2m^2a_\ell(t)+2\widetilde \beta_\ell(t).
\end{equation}

Let $(\mathsf{C},\mathsf{I})$ be a pair of random indices with joint distribution
\begin{equation}
\Pr[\mathsf{C}=c,\mathsf{I}=i]=s_{c}\lambda_{i|c}.
\end{equation}
The random variable
\begin{equation}
\mathsf{Y}\coloneqq-\log_2s_{\mathsf{C}} -2\log_2\lambda_{\mathsf{I}|\mathsf{C}}
\end{equation}
has expectation
\begin{equation}
\mathbb E[\mathsf{Y}]
=H(s_{\Phi,t})+2\sum_{c \in \Xi}s_{c}S(\rho_{c})
=K_{\Phi,t}.
\end{equation}
By Lemma~\ref{lem:uniform-positive}, there exist $\widetilde t_1\in(0,t']$ and $\widetilde \kappa \in (0,1/2]$, depending only on $\Phi$, such that
\begin{equation}
s_{c}\lambda_{i|c}\geq\widetilde \kappa
\end{equation}
for every pair $(c,i)$ and every $t\in[0,\widetilde t_1]$.
Since $0<s_{c}\leq1$,
\begin{equation}
s_{c}\lambda_{i|c}^2
=\frac{(s_{c}\lambda_{i|c})^2}{s_{c}}
\geq\widetilde\kappa^2.
\end{equation}
Thus
\begin{equation}
0\leq \mathsf{Y}\leq L\coloneqq2\log_2(1/\widetilde\kappa).
\end{equation}
Thus, for independent copies $\mathsf{Y}_{1},\ldots,\mathsf{Y}_{\ell}$, 
Hoeffding's inequality~\cite{GrimmettStirzaker2001} gives
\begin{equation}
\begin{aligned}
\widetilde \beta_\ell(t)
&=\Pr\!\left[\frac1\ell\sum_{j=1}^{\ell}\mathsf{Y}_j
<K_{\Phi,t}-\delta\right]\\
&\leq\exp\!\left(-\frac{2\delta^2}{L^2}\ell\right).
\end{aligned}
\end{equation}
Taking $\widetilde v\coloneqq2\delta^2/L^2$ proves the claim.
\end{proof}

\begin{remark}
In the proofs of Lemmas~\ref{lem:memory-bound}
and~\ref{lem:assisted-overlap}, 
positive definiteness of the individual source states is actually not required: positive definiteness of the average state, together with continuity of the perturbed average state at zero, suffices for the proofs.
The same overlap bounds also hold when redundant subsystems
are retained, provided that the block purifications include
these subsystems.
In this case, however, the entropy quantities appearing
in the bounds also include the redundant contributions.
We retain the reduced-source assumptions throughout
to apply Lemma~\ref{lem:fidelity-comparison} and to identify
the resulting strong-converse rate thresholds with the optimal compression rates.
\end{remark}

Combining Lemma~\ref{lem:fidelity-comparison}
and Lemma~\ref{lem:assisted-overlap} proves Theorem~\ref{thm:assisted}
by an argument similar to that used in the proof of Theorem~\ref{thm:main}.
\begin{proof}[Proof of Theorem~\ref{thm:assisted}]
    See Appendix~\ref{appendix:proof_assisted}.
\end{proof}

\subsection{Strong converse for qubit-ebit rate region}
\label{subsec:strong_converse_rate_region}
Now, combining the exponential strong converses for the unassisted and unlimited-assistance cases gives a strong converse for the complete qubit-ebit region.
Here, we take the shared resource to be an initial maximally entangled state of Schmidt rank $d_n$, with ebit rate $e_n=\log_2d_n/n$, as in Theorem~\ref{thm:optimal_blind_compression_rate}.
The unlimited-assistance strong converse constrains the quantum rate alone; the unassisted strong converse will additionally constrain the sum of the quantum and entanglement rates.

\begin{corollary}[Exponential converse outside the rate region]
\label{cor:rate-region}
Let $\Phi$ be a quantum ensemble on system $A$.
For every fixed nonnegative rate pair $(r,e)$ satisfying
\begin{equation}
\label{eq:outside-region}
r<R_{\mathrm{blind},\Phi}^{\mathrm{EA}}
\quad\text{or}\quad
r+e<R_{\mathrm{blind},\Phi},
\end{equation}
there exist $\alpha>0$ and a positive integer $n_0$, depending only
on $\Phi,r,e$, such that, for every $n\geq n_0$, every code with an initial maximally entangled resource of Schmidt rank $d_n$ and rates
\begin{equation}
r_n\coloneqq\frac1n\log_2\dim M_n\leq r,
\qquad
e_n\coloneqq\frac1n\log_2d_n\leq e
\end{equation}
satisfies
\begin{equation}
\label{eq:rate_region_fidelity_decay}
F_n(\Phi,\cT_n)\leq2^{-\alpha n}.
\end{equation}
\end{corollary}
\begin{proof}
If
\begin{equation}
    r<R_{\mathrm{blind},\Phi}^{\mathrm{EA}},
\end{equation}
the conclusion follows directly from Theorem~\ref{thm:assisted}.
We therefore only need to consider the second condition,
\begin{equation}
    r+e<R_{\mathrm{blind},\Phi}.
\end{equation}
We use the unassisted converse (Theorem~\ref{thm:main}) in this case. 
To this end, we show that every assisted code in the statement can be converted into an unassisted code with the same fidelity and a communication rate equal to the sum of its original qubit and ebit rates.

Let $(\cE_n, \cD_n)$ be such an assisted code.
In the original protocol, Alice and Bob share $\Upsilon_{d_n}$ before communication.
We consider the following unassisted simulation of this code. 
Alice instead prepares the shared state $\Upsilon_{d_n}$ entirely on her side, initially holding both $A_0$ and $B_0$.
She applies the encoder $\cE_n$ to the source $A^n$ and $A_0$, leaving $B_0$ untouched, and sends both the resulting message $M_n$ and $B_0$ to Bob. Bob then applies the decoder $\cD_n$.
This is an unassisted protocol: the auxiliary state is prepared locally, independently of the source, and the transmission of $B_0$ is included in the communication cost.

More explicitly, the new message system is
$\widetilde M_n\coloneqq M_nB_0$, and the encoding and decoding channels are
\begin{equation}
    \widetilde{\cE}_n(X)\coloneqq(\cE_n\otimes\id_{B_0})\bigl(X\otimes\Upsilon_{d_n}\bigr),
    \qquad
    \widetilde{\cD}_n\coloneqq\cD_n.
\end{equation}
The new encoder is a CPTP map from $A^n$ to $M_nB_0$, obtained by composing local state preparation with the original encoder.
For every input $X$,
\begin{equation}
    \widetilde{\cT}_n(X)\coloneqq\bigl(\widetilde{\cD}_n\circ\widetilde{\cE}_n\bigr)(X)=\cT_n(X).
\end{equation}
The simulation therefore preserves the effective source channel, and hence the compression fidelity, exactly.

The extra transmitted system $B_0$ has dimension $d_n$;
Alice's system $A_0$ is used by the encoder and need not be sent.
Consequently, the communication rate of the simulated code is
\begin{equation}
    \widetilde r_n
    =\frac1n\log_2\dim\widetilde M_n
    =\frac1n\log_2\bigl((\dim M_n)d_n\bigr)
    =r_n+e_n\leq r+e.
\end{equation}
We can now apply Theorem~\ref{thm:main} at the fixed rate bound $r+e$.
Since this bound is strictly below $R_{\mathrm{blind},\Phi}$,
there exist $\alpha>0$ and $n_0$, depending only on $\Phi$ and $r+e$, such that, for all $n\geq n_0$,
\begin{equation}
    F_n(\Phi,\cT_n) =F_n(\Phi,\widetilde{\cT}_n) \leq 2^{-\alpha n}.
\end{equation}
The original assisted code was arbitrary, so this establishes the claimed exponential converse under the second alternative as well.
\end{proof}

Together with the known achievability theorem, this identifies the region~\eqref{eq:rate_region} as the strong-converse rate region.

\section{Conclusion}
\label{sec:conclusion}
We have established exponential strong converses for blind compression of finite-dimensional mixed-state sources under a global squared-fidelity criterion.
In the unassisted setting, the strong-converse threshold equals the optimal Koashi--Imoto rate: the entropy of the average state after removal of the redundant subsystems.
At every rate below this threshold, the compression fidelity decays exponentially with the block length, uniformly over arbitrary encoding and decoding channels.
With unlimited entanglement assistance, we have proved an exponential strong converse at the optimal assisted quantum rate.
The bound is uniform over all finite-dimensional shared states independent of the source, even when their dimensions grow arbitrarily with the block length.
Combining the assisted quantum-rate bound with the unassisted converse yields exponential decay outside the complete achievable qubit-ebit region for protocols supplied with an initial maximally entangled resource.
Together with the known achievability theorem by Baghali Khanian and Winter, this identifies the entire region as a strong-converse rate region: allowing any fixed positive global error tolerance below one leaves the region unchanged.

The common ingredient is the KI overlap, which quantifies preservation of the essential information specified by the KI decomposition.
For unassisted protocols, the overlap is controlled by the dimension of the compressed system.
For assisted protocols, a message-dimension domination relation gives a bound independent of the size and form of the shared resource.
Then, we convert bounds on this overlap into bounds on the compression fidelity.

This work suggests several directions for future research.
Our results establish positive decay exponents, yet we do not determine the optimal strong-converse exponents.
Characterizing their dependence on the source and the rates remains an open problem.
It would also be interesting to obtain finite-block-length bounds on the communication rate required to achieve a target fidelity for a given entanglement budget.
Such bounds would quantify the difference between the rate
required at a finite block length and the asymptotic optimal rate.
In addition, it remains to determine whether, and how rapidly, the fidelity vanishes when the rate pairs approach the boundary of the achievable region from outside as the block length increases.
More broadly, the KI-overlap method may be useful for other information-processing tasks in which exact source structure must be related to approximate information preservation.

\begin{acknowledgments}
This work grew out of earlier discussions in 2019 and 2020 with Anurag Anshu, Felix Leditzky, and Debbie Leung about strong converses for quantum data compression.
At the time, we aimed to use the R{\'e}nyi-entropy method introduced in Ref.~\cite{Leditzky2016}, but we could not find a way to apply the technique to blind compression of mixed states.
Although we eventually set the problem aside, Zahra Baghali Khanian's recent progress on strong converses for quantum data compression~\cite{Khanian2022} inspired me to revisit the problem.

I am deeply grateful to Zahra Baghali Khanian for explaining their results and for stimulating discussions.
I thank Anurag Anshu and Felix Leditzky for discussions and ideas from our previous project.
I also thank Debbie Leung for providing constructive suggestions and advice.
The author is supported by a Mike and Ophelia Lazaridis Fellowship, the Funai Foundation, and a Perimeter Residency Doctoral Award.
\end{acknowledgments}

\textbf{Use of artificial intelligence.}
The author used OpenAI's GPT-6 Astra to explore proof ideas,
check proofs, and polish the language.
The author had previously established strong converse results
for classical ensembles (ensembles of mutually commuting states) and had also developed the basic idea of the KI overlap.
During the development of the more general results presented here, the central idea underlying the proof of Lemma~\ref{lem:gap} emerged through discussions with GPT-6 Astra.
The author has verified all mathematical results and arguments and takes full responsibility for the content of the paper.

\bibliographystyle{apsrmp4-2}
\bibliography{strong_converse_blind_compression}

\appendix

\section{General sources with a quantum reference}
\label{appendix:quantum-reference}
In this appendix, we extend the strong converse results to general mixed-state sources with a quantum reference considered in Ref.~\cite{Khanian2020b}, by reducing such sources to ensembles.

We consider a finite-dimensional source $\omega^{AZ}$, where
only $A$ is compressed and $Z$ is an inaccessible quantum reference.
The fidelity criterion requires preservation of the joint state, including its correlations with $Z$. 
On the support of $\omega^A$, write the KI decomposition~\cite{Koashi2002,Khanian2020b}, up to a reordering of tensor factors, as

\begin{equation}
(\Gamma_\omega\otimes I_Z)\omega^{AZ}
(\Gamma_\omega^\dagger\otimes I_Z)
=\bigoplus_c s_{\omega,c}\,
\omega_c^{Q_cZ}\otimes\rho_{R_c}^{(c)}.
\end{equation}
The block probabilities $s_{\omega,c}$ are positive, each
$\omega_c^{Q_cZ}$ is normalized, and $\rho_{R_c}^{(c)}$ is the redundant state, as in Sec.~\ref{subsec:structure_quantum_ensemble}. 
Define the redundancy-free state as
\begin{equation}
\widehat\omega\coloneqq\bigoplus_c s_{\omega,c}\omega_c^{Q_c}, 
\end{equation}
and set 
\begin{equation}
R_{\mathrm{blind},\omega}\coloneqq S(\widehat\omega),
\qquad
R_{\mathrm{blind},\omega}^{\mathrm{EA}}
\coloneqq S(\widehat\omega)-\tfrac12H(s_\omega).
\end{equation}
These are the rate thresholds of the general-source coding
theorem~\cite{Khanian2020b}. 
To transfer the exponential strong converse results,
we use the informationally complete measurement reduction of
Ref.~\cite[Sec.~III]{Khanian2020b}.
Informational completeness ensures that the resulting conditional ensemble has the same KI structure and rate thresholds as the joint source. 
Fidelity monotonicity then bounds the source fidelity by the
ensemble fidelity. 

\begin{corollary}[General mixed-state sources]
\label{cor:general-reference}
Theorems~\ref{thm:main} and~\ref{thm:assisted}, together with
Corollary~\ref{cor:rate-region}, hold for $\omega^{AZ}$ with the thresholds defined above and the squared global fidelity
\begin{equation}
F_n(\omega,\cT_n)\coloneqq F\!\left(
(\omega^{AZ})^{\otimes n},
(\cT_n\otimes\id_{Z^n})((\omega^{AZ})^{\otimes n})\right).
\end{equation}
\end{corollary}
\begin{proof}
Choose a finite informationally complete POVM $\{M_z\}_z$ on $Z$, and define its unnormalized conditional operators and probabilities by
\begin{equation}
\tau_z\coloneqq\Tr_Z[(I_A\otimes M_z)\omega^{AZ}],
\qquad p_z\coloneqq\Tr\tau_z.
\end{equation}
Discard outcomes with $p_z=0$, and form the finite ensemble
\begin{equation}
\Phi_\omega\coloneqq\{p_z,\tau_z/p_z\}_{z:p_z>0}.
\end{equation}
Its average is $\omega^A$. 
Since $\{M_z\}_z$ is informationally complete, there exists a set $\{D_z\}_z$ of Hermitian operators on $Z$ such that for any operator $V$ on $Z$, 
\begin{equation}
    V = \sum_{z} \tr[M_z V]D_z. 
\end{equation}
Therefore, we have 
\begin{equation}
\label{eq:reconstruction}
\omega^{AZ}=\sum_z\tau_z\otimes D_z.
\end{equation}
Note that $p_z=0$ implies $\tau_z=0$, so the discarded outcomes do not contribute to~\eqref{eq:reconstruction}.
In particular, a channel $\mathcal N$ on $A$ preserves $\omega^{AZ}$ if and only if
it preserves all the conditional operators:
\begin{equation}
(\mathcal N\otimes\id_Z)(\omega^{AZ})=\omega^{AZ}
\quad\Longleftrightarrow\quad
\mathcal N(\tau_z)=\tau_z\quad\text{for every }z.
\end{equation}

This procedure also preserves the KI decomposition. 
The redundant part of $\omega^{AZ}$ appears in every $\tau_z$. Conversely, the KI structure 
\begin{equation}
\tau_z=\bigoplus_c A_{z,c}\otimes\rho_{R_c}^{(c)}
\end{equation}
gives, by~\eqref{eq:reconstruction}, 
\begin{equation}
\omega^{AZ}=\bigoplus_c W_c^{Q_cZ}\otimes\rho_{R_c}^{(c)},
\qquad
W_c^{Q_cZ}\coloneqq\sum_z A_{z,c}\otimes D_z.
\end{equation}
Consequently, the original source and the resulting ensemble have the same maximal KI structure. 
Since their average state on $A$ is also the same, both optimal rate thresholds $R_{\mathrm{blind},\omega}$ and $R_{\mathrm{blind},\omega}^{\mathrm{EA}}$ are preserved under this reduction.

Let $X$ be a (classical) system, 
and define a measurement channel $\mathcal{M}:Z \to X$ by 
\begin{equation}
    \mathcal{M}(V) \coloneqq \sum_z \tr[M_zV]\ketbra{z}{z}_X. 
\end{equation}
Since $\mathcal M^{\otimes n}$ acts only on the reference
systems and the effective protocol channel $\cT_n$ acts only on $A^n$, these channel commute with each other. 
Hence, monotonicity of fidelity gives 
\begin{equation}
F_n(\omega,\cT_n)\leq F_n(\Phi_\omega,\cT_n).
\end{equation}
Applying the ensemble results to the right-hand side proves all three strong converse results. 
\end{proof}

\section{Proof of Lemma~\ref{lem:fidelity-comparison}}
\label{appendix:proof_lem_comparison}
Throughout this appendix, let $\Phi=\{p_x,\rho_x\}_{x\in\Sigma}$ be a redundancy-free ensemble with $p_x>0$ and $\rho_x>0$ for every $x\in\Sigma$, as discussed in Sec.~\ref{subsec:reduction}.

First, we record a consequence of the KI structure theorem~\cite[Theorem~3]{Koashi2002} that characterizes channels that preserve every state of a redundancy-free ensemble. 

\begin{lemma}[Redundancy-free structure of ensemble-preserving operations]
\label{lem:redundancy_free_KI_preserving_channel}
Let $\Phi = \{p_x,\rho_x\}_{x \in \Sigma}$ be a redundancy-free quantum ensemble on system $A$.
Suppose that a channel $\mathcal{N}:A\to A$ satisfies
\begin{equation}
    \mathcal{N}(\rho_x) = \rho_x \quad \forall x\in\Sigma.
\end{equation}
Then, for any choice of Kraus operators $\{N_a\}_a$ of $\mathcal{N}$, we have
\begin{equation}
    N_a \in \Span\{P_c\}_{c \in \Xi} \quad \forall a.
\end{equation}
\end{lemma}
\begin{proof}
    Fix an arbitrary Kraus representation
    \begin{equation}
        \mathcal N(X)=\sum_aN_aXN_a^\dagger
    \end{equation}
    and construct a corresponding Stinespring isometry
    \begin{equation}
        U_{\mathcal N} = \sum_aN_a\otimes \ket{a}_E.
    \end{equation}
    The KI structure theorem for operations that preserve the
    ensemble~\cite[Theorem~3]{Koashi2002} implies that, for each block $c$, there is a unit vector $\ket{e_c}_E \in \cH_E$ such that
    \begin{equation}
        U_{\mathcal N}\ket{\psi} = \ket{\psi}\otimes \ket{e_c}_E
        \quad\text{for every }\ket{\psi}\in \cH_{Q_c}.
    \end{equation}
    Consequently,
    \begin{equation}
        N_aP_c=\langle a|e_c\rangle P_c.
    \end{equation}
    By taking the sum over $c$, we have
    \begin{equation}
        N_a= N_a\sum_{c \in \Xi} P_c = \sum_{c \in \Xi}\langle a|e_c\rangle P_c \in \Span\{P_c\}_{c \in \Xi} \quad \forall a,
    \end{equation}
  which completes the proof.
\end{proof}

Now, to prove Lemma~\ref{lem:fidelity-comparison}, we introduce the following operators: 
 \begin{align}
 \label{eq:def_Hx}
 H_x&\coloneqq \ln\rho_x-\ln\rho_\Phi,\\
 \label{eq:def_L}
 L_\Phi&\coloneqq I\otimes S_\Phi^{\Trans}-\sum_{x \in \Sigma}p_xH_x\otimes\rho_x^{\Trans}.
 \end{align}

The operator $L_\Phi$ describes the first-order loss of Holevo information under a weak noise process. 
The key question is which perturbations have vanishing first-order information loss.
Lemma~\ref{lem:gap} characterizes the operators generating these perturbations as elements of the linear span of the KI projectors for a redundancy-free ensemble.
This kernel characterization yields a spectral gap away from the KI subspace and, in turn, the quantitative fidelity bound.

More precisely, for an arbitrary operator $V$, consider the channel
\begin{equation}
    \mathcal N_{\varepsilon,V}(\rho)
=
\sqrt{I-\varepsilon V^\dagger V}\,\rho\,
\sqrt{I-\varepsilon V^\dagger V}
+\varepsilon V\rho V^\dagger
\end{equation}
for sufficiently small $\varepsilon\geq 0$.
Writing $\Phi_{\varepsilon,V}
=\{p_x,\mathcal N_{\varepsilon,V}(\rho_x)\}$ and recalling
\begin{equation}
   \chi_\Phi \coloneqq \Tr S_\Phi
   \,\,\,\, \textrm{where}
   \,\,\,\,
    S_\Phi \coloneqq \sum_{x \in \Sigma}p_x\rho_x\ln\rho_x-\rho_\Phi\ln\rho_\Phi,
\end{equation}
we have
\begin{equation}
    \chi_\Phi-\chi_{\Phi_{\varepsilon,V}}=\varepsilon\langle\!\langle V|L_\Phi|V\rangle\!\rangle +O(\varepsilon^2).
\end{equation}
Note that $\chi_\Phi$ is the Holevo information of $\Phi$ in natural-logarithm units:
\begin{equation}
\chi_\Phi=(\ln2)\left(S(\rho_\Phi)-\sum_{x \in \Sigma}p_xS(\rho_x)\right).
\end{equation}
Therefore, if $L_\Phi \geq 0$, this operator $L_\Phi$ quantifies how much information is lost through the channel $\mathcal N_{\varepsilon,V}$.
In Lemma~\ref{lem:gap}, we show that $L_\Phi \geq 0$ actually holds,
and also prove that the first-order coefficient $\langle\!\langle V|L_\Phi|V\rangle\!\rangle$ vanishes if and only if $\vk{V}$ belongs to the span of the vectorized KI projectors.
\begin{lemma}[Properties of $L_\Phi$]\label{lem:gap}
For $L_\Phi$ defined in~\eqref{eq:def_L}, we have
\begin{enumerate}
    \item[(i)] Positivity:
    \begin{equation}
        \label{eq:positivity_L}
        L_\Phi\geq 0;
    \end{equation}
    \item[(ii)] Kernel:
    \begin{equation}
     \label{eq:kernel_L}
        \ker L_\Phi = \Span\{\vk{P_c}\}_{c \in \Xi}.
    \end{equation}
\end{enumerate}
\end{lemma}

\begin{proof}
\textbf{Proof of (i)}:
We show that $\langle\!\langle V|L_\Phi|V\rangle\!\rangle \geq 0$ for any $V$.

Let $V$ be an arbitrary operator.
For a positive-definite operator $\tau>0$ and a real number $u>0$, define the real-valued function
\begin{equation}
f_{\tau,u}(Z)\coloneqq 2\operatorname{Re}\Tr(Z^\dagger V\tau)-\Tr(Z^\dagger \tau Z)-u\Tr(Z^\dagger Z \tau),
\label{eq:variation}
\end{equation}
and its maximum
\begin{equation}
    m_\tau(u)\coloneqq \max_Z f_{\tau,u}(Z).
\end{equation}
Following the arguments in \cite[Section~9.1]{BoydVandenberghe2004}, we see that the function $f_{\tau,u}$ is strongly concave in $Z$,
and that its unique maximizer, denoted $Z_\tau(u)$, satisfies
\begin{equation}
    \tau Z_\tau(u)+uZ_\tau(u)\tau=V\tau.
\label{eq:stationarity}
\end{equation}
See Appendix~\ref{appendix:quadratic-variational-condition} for details. 

Consider an orthonormal basis in which $\tau$ is diagonal: $\tau=\operatorname{diag}(\tau_i)$.
From~\eqref{eq:stationarity}, it follows that
\begin{equation}
\bigl(Z_\tau(u)\bigr)_{ij}=\frac{\tau_jV_{ij}}{\tau_i+u\tau_j}.
\label{eq:entries}
\end{equation}
Substituting this entrywise expression into~\eqref{eq:variation} gives
\begin{equation}
    m_\tau(u)=\sum_{i,j}\frac{\tau_j^2}{\tau_i+u\tau_j}|V_{ij}|^2.
    \label{eq:entries_m}
\end{equation}

Now, define
\begin{equation}
    g_V(\tau) \coloneqq \Tr(V^\dagger V \tau\ln \tau)
       -\Tr\bigl(V^\dagger(\ln \tau)V \tau\bigr).
\end{equation}
We next establish the integral representation
\begin{equation}
g_V(\tau)=\int_0^\infty
\left[m_\tau(u)-\frac{\Tr(V^\dagger V\tau)}{1+u}\right]du.
\label{eq:integral}
\end{equation}
In the same eigenbasis of $\tau$, we have
\begin{align}
    \Tr(V^\dagger V \tau\ln \tau)
    &=\sum_{i,j} \tau_j|V_{ij}|^2\ln \tau_j,\\
    \Tr\bigl(V^\dagger(\ln \tau)V \tau\bigr)
    &=\sum_{i,j}\tau_j|V_{ij}|^2\ln \tau_i.
\end{align}
Therefore,
\begin{equation}
    g_V(\tau) =\sum_{i,j}\tau_j|V_{ij}|^2 \ln\frac{\tau_j}{\tau_i}.
\end{equation}
On the other hand,~\eqref{eq:entries_m} gives
\begin{equation}
    m_\tau(u)
    =\sum_{i,j}\frac{\tau_j^2}{\tau_i+u\tau_j}|V_{ij}|^2
    =\sum_{i,j}\frac{\tau_j}{u+\tau_i/\tau_j}|V_{ij}|^2.
\end{equation}
Since
\begin{equation}
    \Tr(V^\dagger V\tau)=\sum_{i,j}\tau_j|V_{ij}|^2,
\end{equation}
we obtain
\begin{equation}
    m_\tau(u) -\frac{\Tr(V^\dagger V\tau)}{1+u}=\sum_{i,j}\tau_j|V_{ij}|^2\left(\frac{1}{u+\tau_i/\tau_j}-\frac{1}{u+1}\right).
\end{equation}
Now, for any $r,R>0$,
\begin{equation}
    \int_0^R\left(\frac{1}{u+r}-\frac{1}{u+1}\right)du
    =\ln\frac{R+r}{R+1}-\ln r \xrightarrow[R \to \infty]{} -\ln r.
\end{equation}
Setting $r=\tau_i/\tau_j$ gives
\begin{equation}
    \int_0^\infty\left(\frac{1}{u+\tau_i/\tau_j}-\frac{1}{u+1}\right)du =\ln\frac{\tau_j}{\tau_i}.
\end{equation}
Note that the integral is absolutely convergent: for each $r>0$,
\begin{equation}
    \frac{1}{u+r}-\frac{1}{u+1} =\frac{1-r}{(u+r)(u+1)},
\end{equation}
which is bounded near $u=0$ and is $O(u^{-2})$ as
$u\to\infty$.
Therefore,
\begin{equation}
    \begin{aligned}
        \int_0^\infty\left[m_\tau(u)-\frac{\Tr(V^\dagger V\tau)}{1+u}\right]du
        &=
        \sum_{i,j}\tau_j|V_{ij}|^2
        \int_0^\infty
        \left(
        \frac{1}{u+\tau_i/\tau_j}-\frac{1}{u+1}
        \right)du\\
        &=
        \sum_{i,j}\tau_j|V_{ij}|^2\ln\frac{\tau_j}{\tau_i}\\
        &=g_V(\tau),
    \end{aligned}
\end{equation}
which is~\eqref{eq:integral}.

We now rewrite $\langle\!\langle V|L_\Phi|V\rangle\!\rangle$ using $g_V$.
Applying~\eqref{eq:vec_braket} to the definition~\eqref{eq:def_L},
we obtain
\begin{equation}
    \label{eq:q}
    \begin{aligned}
         &\langle\!\langle V|L_\Phi|V\rangle\!\rangle\\
         &=\Tr(V^\dagger VS_\Phi) -\sum_{x\in\Sigma} p_x\Tr(V^\dagger H_xV\rho_x) \\
         &=\sum_{x\in\Sigma} p_x\Tr(V^\dagger V\rho_x\ln\rho_x)-\Tr(V^\dagger V\rho_{\Phi}\ln\rho_\Phi)-\sum_{x\in\Sigma} p_x\Tr\bigl(V^\dagger(\ln\rho_x)V\rho_x\bigr)+\sum_{x\in\Sigma} p_x\Tr\bigl(V^\dagger(\ln\rho_\Phi)V\rho_x\bigr)\\
         &=\sum_{x\in\Sigma} p_x\underbrace{\left[\Tr(V^\dagger V\rho_x\ln\rho_x)-\Tr\bigl(V^\dagger(\ln\rho_x)V\rho_x\bigr)\right]}_{g_V(\rho_x)}-\underbrace{\left[\Tr(V^\dagger V\rho_\Phi\ln\rho_\Phi)-\Tr\bigl(V^\dagger(\ln\rho_\Phi)V\rho_\Phi\bigr)\right]}_{g_V(\rho_\Phi)} \\
        &= \sum_{x\in\Sigma} p_x g_V(\rho_x)-g_V(\rho_\Phi).
    \end{aligned}
\end{equation}

Combining \eqref{eq:integral} and \eqref{eq:q} therefore gives
\begin{equation}
\vbk{V|L_\Phi|V} = \int_0^\infty
\left[\sum_{x \in \Sigma}p_xm_{\rho_x}(u)-m_{\rho_\Phi}(u)\right]du.
\label{eq:positive}
\end{equation}
Since, for fixed $u$ and $Z$, the function $f_{\tau,u}(Z)$ is linear in $\tau$,
\begin{equation}
m_{\rho_\Phi}(u)
=\max_Z\sum_{x \in \Sigma}p_xf_{\rho_x,u}(Z)
\leq\sum_{x \in \Sigma}p_xm_{\rho_x}(u),
\label{eq:convexity}
\end{equation}
which shows that
\begin{equation}
    \vbk{V|L_\Phi|V} \geq 0.
\end{equation}
Since $V$ is arbitrary, this proves
\begin{equation}
    L_{\Phi}\geq0.
\end{equation}

\textbf{Proof of (ii)}:
We first show $\ker L_\Phi \subseteq \Span\{\vk{P_c}\}_{c \in \Xi}$.
Take an arbitrary linear operator $V$ such that $\vk{V}\in\ker L_\Phi$, i.e., $\vbk{V|L_\Phi|V} = 0$.
The integrand in~\eqref{eq:positive} is continuous and nonnegative, so
\begin{equation}
    \sum_{x \in \Sigma}p_xm_{\rho_x}(u)-m_{\rho_\Phi}(u) = 0
\end{equation}
for every $u > 0$.
By definition,
\begin{equation}
    \sum_{x \in \Sigma}p_x\bigl[m_{\rho_x}(u) -f_{\rho_x,u}(Z_{\rho_{\Phi}}(u))\bigr] = 0.
\end{equation}
Since each term is nonnegative and $p_x>0$,
\begin{equation}
    m_{\rho_x}(u) -f_{\rho_x,u}(Z_{\rho_{\Phi}}(u)) = 0
\end{equation}
for all $x \in \Sigma$.
Since the maximizer is unique, $Z_{\rho_\Phi}(u)=Z_{\rho_x}(u)$ for every $x \in \Sigma$ and every $u>0$.
In particular,
\begin{equation}
\rho_xZ_{\rho_\Phi}(u)+uZ_{\rho_\Phi}(u)\rho_x=V\rho_x
\label{eq:common}
\end{equation}
for every $x \in \Sigma$ and $u > 0$.

Write the spectral decomposition of $\rho_{\Phi}$ as
\begin{equation}
    \rho_{\Phi}=\sum_{a}\kappa_a \ketbra{i_a}{i_a}.
\end{equation}
Let $\Lambda_{\Phi}$ be the finite set of ratios $\kappa_a/\kappa_b$ of eigenvalues of $\rho_{\Phi}$,
and for $\lambda\in\Lambda_{\Phi}$, define
\begin{equation}
W_\lambda \coloneqq \sum_{\kappa_a/\kappa_b=\lambda}\braket{i_a|V|i_b}\ketbra{i_a}{i_b}.
\label{eq:W}
\end{equation}
By definition, we have
\begin{equation}
    \label{eq:V_Wlambda}
    V=\sum_{\lambda\in\Lambda_\Phi}W_\lambda,
\end{equation}
and the explicit solution in~\eqref{eq:entries} yields
\begin{equation}
Z_{\rho_{\Phi}}(u)=\sum_{\lambda\in\Lambda_\Phi}\frac{W_\lambda}{u+\lambda}.
\label{eq:partialfractions}
\end{equation}
Substituting~\eqref{eq:partialfractions} into~\eqref{eq:common}, and using
$u/(u+\lambda)=1-\lambda/(u+\lambda)$, gives
\begin{equation}
\sum_{\lambda\in\Lambda_\Phi}
\frac{\rho_xW_\lambda-\lambda W_\lambda\rho_x}{u+\lambda}=0
\qquad(u>0).
\label{eq:rational}
\end{equation}
Multiplying by $\prod_{\lambda\in\Lambda_\Phi}(u+\lambda)$ gives
\begin{equation}
    \label{eq:matrix_valued_polynomial}
    \sum_{\lambda\in\Lambda_\Phi}\left(\prod_{\lambda^\prime \in\Lambda_\Phi: \lambda^\prime \neq \lambda}(u+\lambda^\prime)\right)\left(\rho_xW_\lambda-\lambda W_\lambda\rho_x\right) = 0.
\end{equation}
Each matrix entry of the left-hand side of~\eqref{eq:matrix_valued_polynomial} is a polynomial of finite degree in $u$, and~\eqref{eq:matrix_valued_polynomial} shows that it vanishes for every $u>0$.
Since a nonzero polynomial has at most as many roots as its degree~\cite[Section~9.5, Proposition~17]{DummitFoote2004}, every entry is the zero polynomial.
Hence, the left-hand side is identically zero and therefore vanishes for all $u \in \mathbb R$.
For any fixed $\lambda \in \Lambda_\Phi$,
evaluating \eqref{eq:matrix_valued_polynomial} at $u=-\lambda$ therefore gives
\begin{equation}
\rho_xW_\lambda=\lambda W_\lambda\rho_x
\label{eq:intertwining}
\end{equation}
for every $x \in \Sigma$.

Fix $\lambda\in\Lambda_\Phi$ such that $W_\lambda\neq0$, and write $W=W_\lambda$.
We construct a channel that preserves every source state.
Set
\begin{equation}
\widetilde W=W^\dagger W+\lambda^{-1}WW^\dagger.
\end{equation}
Choose $0<\varepsilon<\|\widetilde W\|_\infty^{-1}$ and define
\begin{equation}
K_0=\sqrt{I-\varepsilon \widetilde W},\qquad
K_1=\sqrt{\varepsilon}\,W,\qquad
K_2=\sqrt{\varepsilon/\lambda}\,W^\dagger.
\label{eq:kraus}
\end{equation}
Since $\sum_{a=0}^2K_a^\dagger K_a=I$, these matrices define a channel
$\mathcal{N}_W(X)=\sum_{a=0}^2K_aXK_a^\dagger$.
For any source state $\rho=\rho_x$, taking the adjoint of \eqref{eq:intertwining} gives
\begin{equation}
    \rho W^\dagger=\lambda^{-1}W^\dagger \rho.
    \label{eq:intertwining_adjoint}
\end{equation}
By combining~\eqref{eq:intertwining} and~\eqref{eq:intertwining_adjoint}, we obtain
\begin{equation}
    [\rho,W^\dagger W]=[\rho,WW^\dagger]=0.
\end{equation}
Thus $\rho$ commutes with $\widetilde W$ and $K_0$, and
\begin{equation}
    W\rho W^\dagger=\lambda^{-1} \rho WW^\dagger,\qquad
    W^\dagger \rho W=\lambda \rho W^\dagger W.
\end{equation}
It follows that
\begin{equation}
    \begin{aligned}
        \mathcal{N}_W(\rho)
        &=\rho(I-\varepsilon \widetilde W)+\varepsilon W\rho W^\dagger +\frac{\varepsilon}{\lambda}W^\dagger \rho W\\
        &=\rho(I-\varepsilon \widetilde W)+\frac{\varepsilon}{\lambda}\rho WW^\dagger+\varepsilon \rho W^\dagger W \\
        &=\rho.
\end{aligned}
\end{equation}
Since this holds for every $x\in\Sigma$, the channel $\mathcal N_W$ preserves every source state.
Lemma~\ref{lem:redundancy_free_KI_preserving_channel} therefore implies $K_1\in\Span\{P_c\}_{c\in\Xi}$, and hence $W_\lambda\in\Span\{P_c\}_{c\in\Xi}$.
Using~\eqref{eq:V_Wlambda}, we obtain
\begin{equation}
    V\in\Span\{P_c\}_{c \in \Xi},
\end{equation}
and thus
\begin{equation}
\ker L_\Phi\subseteq\Span\{\vk{P_c}\}_{c \in \Xi}.
\end{equation}

For the reverse inclusion, observe that every $P_c$ commutes with $\rho_x$, $\rho_{\Phi}$, and their logarithms.
Moreover,
\begin{equation}
    \label{eq:equivalence_phrho_S}
    \sum_{x \in \Sigma}p_xH_x\rho_x = \sum_{x \in \Sigma}p_x\rho_x\ln\rho_x-(\ln\rho_\Phi)\rho_{\Phi}=S_\Phi.
\end{equation}
Consequently,~\eqref{eq:vec_property},~\eqref{eq:def_L} and~\eqref{eq:equivalence_phrho_S} give
\begin{equation}
L_{\Phi}\vk{P_c}
=\biggl|P_cS_\Phi-\sum_{x \in \Sigma}p_xH_xP_c\rho_x\biggr\rangle\!\biggr\rangle
=0
\end{equation}
for any $c \in \Xi$.
That is, $\vk{P_c} \in \ker L_\Phi$ for all $c \in \Xi$, which proves the reverse inclusion
\begin{equation}
    \ker L_\Phi\supseteq\Span\{\vk{P_c}\}_{c \in \Xi}
\end{equation}
and completes the proof.
\end{proof}

Finally, we prove Lemma~\ref{lem:fidelity-comparison} by using Lemma~\ref{lem:gap}.
\begin{proof}[Proof of Lemma~\ref{lem:fidelity-comparison}]
Apply~\eqref{eq:fidelity_trace} with
\begin{align}
    P &= \rho_{X^nA^n} = \sum_{x^n \in \Sigma^n} p_{x^n}\ketbra{x^n}{x^n} \otimes \rho_{x^n}, \\
    Q &= (\id_{X}^{\otimes n} \otimes \cT_n)(\rho_{X^nA^n}) = \sum_{x^n \in \Sigma^n} p_{x^n}\ketbra{x^n}{x^n} \otimes \cT_n(\rho_{x^n}), \\
    Z_{\mathrm{test}} &= \sum_{x^n \in \Sigma^n} \ketbra{x^n}{x^n} \otimes \left(\bigotimes_{i=1}^n e^{tH_{x_i}}\right)
\end{align}
to obtain
\begin{equation}
    F_n \leq\Tr(P Z_{\mathrm{test}}^{-1})\Tr(Q Z_{\mathrm{test}}).
\end{equation}
The product structure of $P$ and $Z_{\mathrm{test}}$ gives
\begin{equation}
\begin{aligned}
\Tr(PZ_{\mathrm{test}}^{-1})
=\sum_{x^n \in \Sigma^n}\prod_{i=1}^n
  \left[p_{x_i}\Tr
  (\rho_{x_i}e^{-tH_{x_i}})\right]
=\left[\sum_{x \in \Sigma}p_x\Tr
  (\rho_xe^{-tH_x})\right]^n \eqqcolon a(t)^n,
\end{aligned}
\end{equation}
where we set
\begin{equation}
    a(t)=\sum_{x \in \Sigma}p_x\Tr(\rho_xe^{-tH_x}).
\end{equation}
For the second trace, we apply~\eqref{eq:choi_relation} with $\mathcal N=\cT_n$:
\begin{equation}
    \Tr[J_{\cT_n}(Y\otimes X^{\Trans})]
    =\Tr[Y\cT_n(X)].
\end{equation}
This implies that
\begin{equation}
    \begin{aligned}
\Tr(QZ_{\mathrm{test}})
&=\sum_{x^n}p_{x^n}\Tr\!\left[
  \cT_n(\rho_{x^n})
  \bigotimes_{i=1}^n e^{tH_{x_i}}\right]\\
&=\Tr\!\left[
  J_{\cT_n}
  \sum_{x^n}p_{x^n}
  \left(\bigotimes_{i=1}^n e^{tH_{x_i}}\right)
  \otimes\rho_{x^n}^{\Trans}\right]\\
&=\Tr[J_{\cT_n}B(t)^{\otimes n}],
\end{aligned}
\end{equation}
where we set
\begin{equation}
    B(t)=\sum_{x \in \Sigma}p_xe^{tH_x}\otimes\rho_x^{\Trans}.
\end{equation}
Combining the two trace expressions yields
\begin{equation}
    \label{eq:Fn_bound}
    F_n\leq a(t)^n
\Tr[J_{\cT_n}B(t)^{\otimes n}].
\end{equation}
We now bound $a(t)$ and $B(t)$ to obtain the desired bound~\eqref{eq:fidelity-comparison}.

To derive a bound on $B(t)$,
observe that the definitions of $L_\Phi$ and $Q_{\Phi,t}$ give
\begin{equation}\label{eq:comparison-taylor}
  B(t)=I\otimes Q_{\Phi,t}^{\Trans}-tL_\Phi+R(t),
\end{equation}
where
\begin{equation}\label{eq:comparison-remainder}
  R(t):=\sum_{x \in \Sigma} p_x\left(e^{tH_x}-I-tH_x\right)\otimes\rho_x^{\Trans}.
\end{equation}
Choose $t_*>0$ and $m_0>0$ such that
\begin{equation}\label{eq:comparison-faithfulness}
  Q_{\Phi,t}\geq m_0I\qquad(0\leq t\leq t_*).
\end{equation}
Such a choice exists because $Q_{\Phi,t=0}=\rho_{\Phi}>0$ and $Q_{\Phi,t}$ is continuous.
Write $h_x:=\opnorm{H_x}$.
Taylor expansion gives
\begin{equation}
\opnorm{e^{tH_x}-I-tH_x}
  \le \frac{t^2}{2}h_x^2e^{t_*h_x}
  \qquad(0\le t\le t_*).
\end{equation}
Consequently,
\begin{equation}\label{eq:comparison-K}
  \opnorm{R(t)}\le Kt^2,\qquad
  K:=\frac12\sum_{x \in \Sigma}p_x\opnorm{\rho_x}h_x^2e^{t_*h_x}.
\end{equation}
Define
\begin{equation}\label{eq:comparison-tilde}
\begin{aligned}
  \widetilde B_{\Phi,t}
    &:=(I\otimes Q_{\Phi,t}^{-\Trans/2})B(t)(I\otimes Q_{\Phi,t}^{-\Trans/2}),\\
  \widetilde L_{\Phi,t}
    &:=(I\otimes Q_{\Phi,t}^{-\Trans/2})L_\Phi(I\otimes Q_{\Phi,t}^{-\Trans/2}),\\
  \widetilde R_{\Phi,t}
    &:=(I\otimes Q_{\Phi,t}^{-\Trans/2})R(t)(I\otimes Q_{\Phi,t}^{-\Trans/2}).
\end{aligned}
\end{equation}
Here $Q_{\Phi,t}^{\Trans/2}$ and $Q_{\Phi,t}^{-\Trans/2}$ denote $(Q_{\Phi,t}^{\Trans})^{1/2}$ and $(Q_{\Phi,t}^{\Trans})^{-1/2}$, respectively.
Then~\eqref{eq:comparison-taylor} becomes
\begin{equation}\label{eq:comparison-tilde-expansion}
  \widetilde B_{\Phi,t}=I-t\widetilde L_{\Phi,t}+\widetilde R_{\Phi,t}.
\end{equation}
By \eqref{eq:comparison-faithfulness} and \eqref{eq:comparison-K},
\begin{equation}
    \opnorm{\widetilde R_{\Phi,t}}
  \leq \opnorm{Q_{\Phi,t}^{-1}}\opnorm{R(t)}
  \leq \frac{K}{m_0}t^2.
\end{equation}
Choose a constant $C_1>0$ such that
\begin{equation}\label{eq:comparison-c}
    C_1\geq\frac{K}{m_0}.
\end{equation}
Note that this constant is independent of $t$, $n$, and $\cT_n$.
Since $\widetilde R_{\Phi,t}$ is Hermitian,
$\widetilde R_{\Phi,t}\leq C_1t^2I$.
Also since $B(t)\geq 0$, we also have $\widetilde B_{\Phi,t}\geq 0$.
It follows that
\begin{equation}\label{eq:comparison-pre-gap}
  0\leq \widetilde B_{\Phi,t}\leq (1+C_1t^2)I-t\widetilde L_{\Phi,t}.
\end{equation}
Now, we bound $\widetilde L_{\Phi,t}$. 
Recall that Lemma~\ref{lem:gap} gives
$\ker L_\Phi=\Span\{\vk{P_c}:c\in\Xi\}$.
Thus,
\begin{equation}\label{eq:comparison-kernel}
\begin{aligned}
\ker\widetilde L_{\Phi,t}
&=(I\otimes Q_{\Phi,t}^{\Trans/2})(\ker L_\Phi)\\
&=\Span\{\vk{P_c\sqrt{Q_{\Phi,t}}}:c\in\Xi\}
=\ran\Pi_{\Phi,t}.
\end{aligned}
\end{equation}
The last equality uses $Q_{\Phi,t}=q_{\Phi,t}\rho_{\Phi,t}$.
The operator $\widetilde L_{\Phi,t}$ is continuous on the compact interval $[0,t_*]$ and has constant rank. If this rank is positive, its smallest positive eigenvalue has a positive minimum on this interval; if the rank is zero, then $\Pi_{\Phi,t}=I$ and the following inequality holds for any $\mu>0$.
Thus, there exists a constant $\mu>0$, independent of $t$, such that
\begin{equation}\label{eq:comparison-tilted-gap}
\widetilde L_{\Phi,t}\geq\mu(I-\Pi_{\Phi,t})
\qquad (0\leq t\leq t_*).
\end{equation}
Write $\Pi^\perp_{\Phi,t}:=I-\Pi_{\Phi,t}$ and $\gamma_t:=1+C_1t^2$. Combining~\eqref{eq:comparison-pre-gap} and \eqref{eq:comparison-tilted-gap} gives
\begin{equation}
    \widetilde B_{\Phi,t}
  \le\gamma_tI-\mu t\Pi^\perp_{\Phi,t}
  =\gamma_t\Pi_{\Phi,t}+(\gamma_t-\mu t)\Pi^\perp_{\Phi,t}.
\end{equation}
Now choose and fix $t_0\in(0,t_*]$ sufficiently small that
\begin{equation}\label{eq:comparison-final-interval}
  C_1t_0^2\leq 1,\qquad \mu t_0\leq 1.
\end{equation}
For $0\leq t\leq t_0$, we have $1\leq \gamma_t\leq 2$ and
\begin{equation}
    \frac{\gamma_t-\mu t}{\gamma_t}
  =1-\frac{\mu t}{\gamma_t}
  \leq 1-\frac{\mu t}{2}
  \leq e^{-\mu t/2}.
\end{equation}
Thus, with
\begin{equation}\label{eq:comparison-b}
  b:=\frac{\mu}{2},
\end{equation}
we obtain
\begin{equation}\label{eq:comparison-operator}
  0\leq \widetilde B_{\Phi,t}
  \leq (1+C_1t^2)I-t\widetilde L_{\Phi,t}
  \leq (1+C_1t^2)\left(\Pi_{\Phi,t}+e^{-bt}(I-\Pi_{\Phi,t})\right).
\end{equation}
Define
\begin{equation}
    \Theta_t:=\Pi_{\Phi,t}+e^{-bt}(I-\Pi_{\Phi,t}).
\end{equation}
Then~\eqref{eq:comparison-operator} implies
\begin{equation}\label{eq:comparison-tensorized}
  \widetilde B_{\Phi, t}^{\otimes n}\leq (1+C_1t^2)^n\Theta_t^{\otimes n}.
\end{equation}
Observe that
\begin{equation}\label{eq:comparison-trace-congruence}
  \Tr\left[J_{\cT_n}B(t)^{\otimes n}\right]
  =\Tr\left[
    (I\otimes Q_{\Phi,t}^{\Trans/2})^{\otimes n}J_{\cT_n}(I\otimes Q_{\Phi,t}^{\Trans/2})^{\otimes n}
    \widetilde B_{\Phi,t}^{\otimes n}
  \right].
\end{equation}
To identify the transformed Choi operator, note that
\begin{equation}
(I\otimes Q_{\Phi,t}^{\Trans/2})\vk{I}
  =\vk{\sqrt{Q_{\Phi,t}}}
  =\sqrt{q_{\Phi,t}}\,\vk{\sqrt{\rho_{\Phi,t}}}.
\end{equation}
The operator $I\otimes Q_{\Phi,t}^{\Trans/2}$ acts only on the reference half, so it commutes with the action of the channel on the other half.
Hence,
\begin{equation}\label{eq:comparison-choi-state}
\begin{aligned}
  (I\otimes Q_{\Phi,t}^{\Trans/2})^{\otimes n}J_{\cT_n}(I\otimes Q_{\Phi,t}^{\Trans/2})^{\otimes n}
  =(\cT_n\otimes\id)\left[
    \left(\vkb{\sqrt{Q_{\Phi,t}}}{\sqrt{Q_{\Phi,t}}}\right)^{\otimes n}
  \right]
  =q_{\Phi,t}^n\Omega_{\Phi,t}(\cT_n).
\end{aligned}
\end{equation}
Substituting~\eqref{eq:comparison-tensorized} and
\eqref{eq:comparison-choi-state} into
~\eqref{eq:comparison-trace-congruence} yields
\begin{equation}\label{eq:comparison-before-scalar}
   \Tr\left[J_{\cT_n}B(t)^{\otimes n}\right]
   \leq q_{\Phi,t}^{\,n}(1+C_1t^2)^n \Tr\left[\Omega_{\Phi,t}(\cT_n)\Theta_t^{\otimes n}\right].
\end{equation}

Next, we bound $a(t)$.
The coefficient of $t$ in the Taylor expansion of $a(t)$ at zero is
\begin{equation}
    \frac{d}{dt}a(t)\Bigg\lvert_{t=0}=-\sum_{x \in \Sigma}p_x\Tr(\rho_xH_x)=-\chi_{\Phi}.
\end{equation}
Repeating the Taylor expansion argument used above gives
\begin{equation}\label{eq:comparison-scalar-remainder}
  a(t)=1-\chi_{\Phi} t+r_a(t),\qquad
  |r_a(t)|\le k_at^2,\qquad
  k_a:=\frac12\sum_{x \in \Sigma}p_xh_x^2e^{t_*h_x}.
\end{equation}
Indeed,
\begin{equation}
      \left|\Tr\left[\rho_x(e^{-tH_x}-I+tH_x)\right]\right|
  \le\opnorm{e^{-tH_x}-I+tH_x}
  \le\frac{t^2}{2}h_x^2e^{t_*h_x}.
\end{equation}
Since $q_{\Phi,t}=1+\chi_{\Phi} t>0$, we have
\begin{equation}
    \begin{aligned}
      a(t)q_{\Phi,t}
      &=(1-\chi_{\Phi} t+r_a(t))(1+\chi_{\Phi} t)\\
      &=1-\chi_{\Phi}^2t^2+(1+\chi_{\Phi} t)r_a(t)\\
      &\le1+C_0t^2,
    \end{aligned}
\end{equation}
where we may take
\begin{equation}\label{eq:comparison-C0}
  C_0:=k_a(1+|\chi_{\Phi}|t_*) \geq 0.
\end{equation}
Consequently,
\begin{equation}
    \label{eq:bound_at}
    \left[a(t)q_{\Phi,t}(1+C_1t^2)\right]^n
  \le\left[(1+C_0t^2)(1+C_1t^2)\right]^n
  \le e^{(C_0+C_1)t^2n}.
\end{equation}
Choose
\begin{equation}\label{eq:comparison-C}
  C:=C_0+C_1>0.
\end{equation}
Combining~\eqref{eq:Fn_bound}, \eqref{eq:bound_at}, and~\eqref{eq:comparison-before-scalar} proves~\eqref{eq:fidelity-comparison} for every $0<t\leq t_0$.
\end{proof}

\section{Proof of the entanglement-assisted strong converse}
\label{appendix:proof_assisted}
We prove Theorem~\ref{thm:assisted} by combining the fidelity bound (Lemma~\ref{lem:fidelity-comparison}) with the assisted KI-overlap bound (Lemma~\ref{lem:assisted-overlap}). 

\begin{proof}[Proof of Theorem~\ref{thm:assisted}]
Fix
\begin{equation}
0\leq r<R_{\mathrm{blind},\Phi}^{\mathrm{EA}},
\qquad
\log_2\dim M_n\leq nr,
\end{equation}
and let $\cT_n$ be the effective source channel of an assisted code with an arbitrary finite-dimensional source-independent resource.

Let $t_0,C,b>0$ be the constants of Lemma~\ref{lem:fidelity-comparison}. 
For parameters $t>0$ and $\eta\in(0,1/2)$ to be fixed below, set
\begin{equation}
\Pi\coloneqq\Pi_{\Phi,t},\qquad
\Pi^\perp\coloneqq I-\Pi,\qquad
\Theta_t\coloneqq\Pi+e^{-bt}\Pi^\perp,\qquad
k\coloneqq\lfloor\eta n\rfloor.
\end{equation}
For each subset $S\subseteq[n]$, define
\begin{equation}
P_S\coloneqq\bigotimes_{j=1}^n
\begin{cases}
\Pi^\perp,&j\in S,\\
\Pi,&j\notin S,
\end{cases}
\qquad
P_{\leq k}\coloneqq\sum_{\substack{S\subseteq[n]\\|S|\leq k}}P_S.
\end{equation}
The projectors $P_S$ are mutually orthogonal, so
\begin{equation}
\begin{aligned}
\Theta_t^{\otimes n}
&=\sum_{S\subseteq[n]}e^{-bt|S|}P_S\\
&\leq P_{\leq k}+e^{-bt(k+1)}I\\
&\leq P_{\leq k}+e^{-bt\eta n}I.
\end{aligned}
\end{equation}
Moreover,
\begin{equation}
P_{\leq k}
\leq\sum_{\substack{S\subseteq[n]\\|S|\leq k}}
I_{A_SR_S}\otimes\Pi^{\otimes S^c}.
\end{equation}
Here tensor factors are reordered and placed in the indicated
systems. 
Define
\begin{equation}
g_S\coloneqq\Tr\!\left[
\Omega_{\Phi,t}(\cT_n)
\bigl(I_{A_SR_S}\otimes\Pi^{\otimes S^c}\bigr)\right].
\end{equation}
For $0<t\leq t_0$, Lemma~\ref{lem:fidelity-comparison} then gives
\begin{equation}
F_n(\Phi,\cT_n)
\leq e^{Ct^2n}\left(
\sum_{\substack{S\subseteq[n]\\|S|\leq k}}g_S
+e^{-bt\eta n}\right).
\end{equation}

Fix a subset $S\subseteq[n]$ with $|S|\leq k$ and put $\ell=n-|S|$.
Then $\ell\geq n(1-\eta)>0$, so the assisted overlap lemma
applies at block length $\ell$. Define the reduced source channel
\begin{equation}
\cT_{n,S}(X_{A_{S^c}})
\coloneqq\Tr_{A_S}\cT_n\!\left(
X_{A_{S^c}}\otimes\rho_{\Phi,t}^{\otimes S}\right).
\end{equation}
This is again an assisted protocol with the same message and shared state: Alice prepares the omitted inputs locally, and Bob traces out their outputs. 
Hence, 
\begin{equation}
\Tr_{A_SR_S}\Omega_{\Phi,t}(\cT_n)
=(\cT_{n,S}\otimes\id_{R_{S^c}})
(\psi_{\Phi,t}^{\otimes\ell}), 
\end{equation}
and therefore
\begin{equation}
    g_S = g_\ell(\Phi,t,\cT_{n,S}).
\end{equation}
Let $\widetilde t_1>0$ be a constant supplied by Lemma~\ref{lem:assisted-overlap}.
Then, for every $\delta > 0$, there exists a constant $\widetilde v > 0$ such that 
\begin{equation}
g_S\leq2\,2^{2nr-\ell(K_{\Phi,t}-\delta)}
+2e^{-\widetilde v\ell}
\end{equation}
for every $t\in(0,\widetilde t_1]$.
If $K_{\Phi,t}-\delta>0$, this is at most
\begin{equation}
2\,2^{2nr-n(1-\eta)(K_{\Phi,t}-\delta)}
+2e^{-\widetilde v n(1-\eta)}.
\end{equation}
Using
\begin{equation}
\sum_{s=0}^{\lfloor\eta n\rfloor}\binom ns
\leq2^{nh_2(\eta)},
\end{equation}
we obtain 
\begin{equation}
\label{eq:assisted-central}
\begin{aligned}
F_n(\Phi,\cT_n)
\leq e^{Ct^2n}\Bigl\{
2^{nh_2(\eta)}\bigl[
2\,2^{2nr-n(1-\eta)(K_{\Phi,t}-\delta)}
+2e^{-\widetilde v n(1-\eta)}\bigr]
+e^{-bt\eta n}\Bigr\}
\end{aligned}
\end{equation}
when $K_{\Phi,t}-\delta>0$. 

We choose the constants similarly to the unassisted proof. 
First set
\begin{equation}
\Delta\coloneqq K_{\Phi,0}-2r>0,
\qquad
\delta\coloneqq\Delta/8.
\end{equation}
Fix the constant $\widetilde v$ for this choice of $\delta$ from Lemma~\ref{lem:assisted-overlap}, and pick $\widetilde t_2 > 0$ so that 
\begin{equation}
K_{\Phi,t}\geq K_{\Phi,0}-\delta
\end{equation}
for every $0 < t \leq \widetilde t_2$. 
In particular, this choice satisfies $K_{\Phi,t}-\delta\geq K_{\Phi,0}-2\delta>0$.
Next choose $\eta\in(0,1/2)$ sufficiently small that
\begin{equation}
\begin{aligned}
A_1&\coloneqq(\ln2)\bigl[
(1-\eta)(K_{\Phi,0}-2\delta)-2r-h_2(\eta)\bigr]>0,\\
A_2&\coloneqq\widetilde v(1-\eta)-(\ln2)h_2(\eta)>0.
\end{aligned}
\end{equation}
Such a choice is possible because, as $\eta\to0^+$, the two expressions tend to $3(\ln2)\Delta/4 > 0$ and $\widetilde v > 0$, respectively. 
Finally choose 
\begin{equation}
0<t\leq\min\{t_0,\widetilde t_1,\widetilde t_2\}
\quad\text{such that}\quad
Ct^2<\frac12\min\{A_1,A_2,bt\eta\}.
\end{equation}
This is possible because $A_1,A_2$ are fixed positive constants,
and $Ct^2/(bt\eta)\to0$ as $t\to 0^+$.

Define
\begin{equation}
\widetilde\alpha\coloneqq
\min\{A_1-Ct^2,A_2-Ct^2,bt\eta-Ct^2\}>0.
\end{equation}
Equation~\eqref{eq:assisted-central} implies
\begin{equation}
\begin{aligned}
F_n(\Phi,\cT_n)
\leq2e^{-n(A_1-Ct^2)}
+2e^{-n(A_2-Ct^2)}
+e^{-n(bt\eta-Ct^2)}
\leq5e^{-\widetilde\alpha n}.
\end{aligned}
\end{equation}
For example, take
\begin{equation}
\alpha\coloneqq\frac{\widetilde\alpha}{2\ln2},
\qquad
n_0\coloneqq\max\left\{1,
\left\lceil\frac{2\ln5}{\widetilde\alpha}\right\rceil\right\}.
\end{equation}
Then every $n\geq n_0$ satisfies
\begin{equation}
F_n(\Phi,\cT_n)
\leq5e^{-\widetilde\alpha n}
\leq e^{-\widetilde\alpha n/2}
=2^{-\alpha n}.
\end{equation}
\end{proof}

\section{Strong concavity and optimality condition}
\label{appendix:quadratic-variational-condition} 
In this appendix, we discuss the properties of the function $f_{\tau,u}$ in~\eqref{eq:variation}. 

Let $\cH$ be a finite-dimensional complex Hilbert space, let $\tau$ be a positive-definite operator on $\cH$, and let $u > 0$ be a positive real number.
For a linear operator $V \in \Lin(\cH)$, consider the
real-valued function
\begin{equation}
    f_{\tau,u}(Z)
    =2\operatorname{Re}\operatorname{Tr}(Z^\dagger V\tau)-\operatorname{Tr}(Z^\dagger\tau Z)-u\operatorname{Tr}(Z^\dagger Z\tau),
    \qquad Z\in \Lin(\cH).
\end{equation}
Here, we show that this function is strongly concave and attains a unique maximum, and that its maximizer satisfies~\eqref{eq:stationarity}. 

We regard $\Lin(\cH)$ as a real vector space equipped with the real Hilbert--Schmidt inner product and its induced norm,
\begin{equation}
    \langle X,Y\rangle_{\mathbb R}=\operatorname{Re}\operatorname{Tr}(X^\dagger Y),
    \qquad
    \|X\|_2^2=\operatorname{Tr}(X^\dagger X).
\end{equation}
For every matrix direction $H$, differentiation along a real scalar parameter gives
\begin{equation}
    \begin{aligned}
    \left.\frac{d^2}{ds^2}f_{\tau,u}(Z+sH)\right|_{s=0}
    &=
    -2\operatorname{Tr}(H^\dagger\tau H)
    -2u\operatorname{Tr}(H^\dagger H\tau)\\
    &\leq -2(1+u)\lambda_{\min}(\tau)\|H\|_2^2, 
\end{aligned}
\end{equation}
where $\lambda_{\min}(\tau) > 0$ denotes the minimum eigenvalue of $\tau$. 
Hence, $f_{\tau,u}$ is strongly concave~\cite[Sec.~9.1.2, pp.~459--460]{BoydVandenberghe2004}. 
Additionally, since this function is continuous and strongly concave on the entire finite-dimensional real vector space of complex matrices, it attains a unique global maximum.

Moreover, the first directional derivative is
\begin{equation}
    \left.\frac{d}{ds}f_{\tau,u}(Z+sH)\right|_{s=0}=2\operatorname{Re}\operatorname{Tr}\bigl[H^\dagger(V\tau-\tau Z-uZ\tau)\bigr].
\end{equation}
The first-order optimality condition requires this derivative to vanish in every matrix direction at the maximizer; it is also sufficient by concavity~\cite[Sec.~9.1, Eq.~(9.2)]{BoydVandenberghe2004}.
Thus, we have 
\begin{equation}
    \tau Z_\tau(u)+uZ_\tau(u)\tau=V\tau, 
\end{equation}
which is~\eqref{eq:stationarity}. 

\section{Miscellaneous lemmas}
This appendix collects auxiliary lemmas used in the proofs.

\begin{lemma}[Uniform positive definiteness]
\label{lem:uniform-positive}
Let $\tau > 0$, and fix $t_{\max} > 0$.
Let $\{\tau_t\}_{0\leq t\leq t_{\max}}$ be a family of states such that $\tau_0 = \tau$ and
\begin{equation}
\lim_{t\to 0^+}\|\tau_t-\tau\|_\infty=0.
\end{equation}
Then there exist $\widetilde t\in(0,t_{\max}]$ and
$\lambda>0$ such that
\begin{equation}
\tau_t\geq\lambda I
\quad\text{for every }t\in[0,\widetilde t].
\end{equation}
In particular, all eigenvalues of these states are uniformly bounded below by $\lambda$.
\end{lemma}
\begin{proof}
Let $\lambda_{\min}(\tau)>0$ be the minimum eigenvalue of $\tau$.
By convergence at zero, choose
$\widetilde t>0$ small enough that
\begin{equation}
\|\tau_t-\tau\|_\infty\leq \lambda_{\min}(\tau)/2
\quad\text{for every }t\in[0,\widetilde t].
\end{equation}
For every such $t$,
\begin{equation}
\tau_t\geq\tau-\|\tau_t-\tau\|_\infty I
\geq(\lambda_{\min}(\tau)/2)I.
\end{equation}
The claim follows with $\lambda=\lambda_{\min}(\tau)/2$.
\end{proof}

\begin{lemma}
    \label{lem:pinching}
    Let $A$ and $B$ be quantum systems.
    Then any positive semidefinite operator $P_{AB}$ on the joint system $AB$ satisfies
    \begin{equation}
        P_{AB}\leq (\dim A) I_A\otimes P_B, 
    \end{equation}
    where $P_B\coloneqq\Tr_A (P_{AB})$.
\end{lemma}
\begin{proof}
    Choose an orthonormal basis $\{\ket{j}_A\}_j$ of $A$ and set
    \begin{equation}
    Q_j\coloneqq|j\rangle\langle j|_A\otimes I_B,
    \qquad 1\leq j\leq \dim A.
    \end{equation}
    For every vector $|\psi\rangle \in \cH_A\otimes \cH_B$, the Cauchy--Schwarz inequality gives
    \begin{equation}
    \begin{aligned}
    \langle \psi|P_{AB}|\psi\rangle
    &=\left\|\sum_{j=1}^{\dim A}P_{AB}^{1/2}Q_j|\psi\rangle\right\|^2\\
    &\leq (\dim A)\sum_{j=1}^{\dim A}\|P_{AB}^{1/2}Q_j|\psi\rangle\|^2\\
    &=(\dim A)\left\langle \psi\middle|
    \sum_{j=1}^{\dim A}Q_jP_{AB}Q_j\middle|\psi\right\rangle.
    \end{aligned}
    \end{equation}
    Thus $P_{AB}\leq (\dim A)\sum_jQ_jP_{AB}Q_j$.
    The diagonal blocks
    \begin{equation}
    P_{jj}\coloneqq(\langle j|_A\otimes I_B)P_{AB}
    (|j\rangle_A\otimes I_B)
    \end{equation}
    are positive and sum to $P_B$, so $P_{jj}\leq P_B$. Consequently,
    \begin{equation}
    \sum_{j=1}^{\dim A}Q_jP_{AB}Q_j
    =\sum_{j=1}^{\dim A}|j\rangle\langle j|_A\otimes P_{jj}
    \leq I_A\otimes P_B,
    \end{equation}
    which proves the claim.
\end{proof}

\end{document}